\documentclass[a4paper,11pt]{article}
\usepackage[utf8]{inputenc}
\usepackage{enumerate}
\usepackage[OT4]{fontenc}
\usepackage{amssymb}
\usepackage{amsmath}
\usepackage{amsfonts}
\usepackage{amsthm}
\usepackage{epstopdf}
\usepackage{todonotes}
\usepackage{algorithm}
\usepackage{algorithmicx}
\usepackage[noend]{algpseudocode}
\usepackage{fullpage}
\usepackage{authblk}
\usepackage{enumitem}

\begin{document}

\title{Improved Bounds for Shortest Paths in Dense Distance Graphs}
\date{}
\author[1]{Paweł Gawrychowski}
\author[2]{Adam Karczmarz\thanks{Supported by the grant NCN2014/13/B/ST6/01811 of the Polish Science Center.}}
\affil[1]{University of Haifa, Israel}
\affil[ ]{\texttt{gawry1@gmail.com}}
\affil[2]{Institute of Informatics, University of Warsaw, Poland}
\affil[ ]{\texttt{a.karczmarz@mimuw.edu.pl}}

\maketitle
\newtheorem{lemma}{Lemma}
\newtheorem{theorem}{Theorem}
\newtheorem{remark}{Remark}
\newtheorem{corollary}{Corollary}

\maketitle

\newtheorem{fact}{Fact}
\newcommand{\fak}{Fakcharoenphol }
\newcommand{\rdiv}{\mathcal{P}}
\newcommand{\region}{\mathcal{p}}
\newcommand{\parti}{\mathcal{S}}
\newcommand{\univ}{\mathcal{U}}
\newcommand{\mmat}{\mathcal{M}}

\newcommand{\bnd}{\partial}
\newcommand{\ddg}{\text{DDG}}

\newcommand{\vals}{\mathbb{R}}

\newcommand{\msize}{m}
\newcommand{\mrows}{R}
\newcommand{\mcols}{C}
\newcommand{\mrowcnt}{k}
\newcommand{\mcolcnt}{l}
\newcommand{\mrow}{r}
\newcommand{\mcol}{c}

\newcommand{\mcurrows}{\bar{\mrows}}
\newcommand{\mcurcols}{\bar{\mcols}}

\newcommand{\minit}{\textsc{Init}}
\newcommand{\marow}{\textsc{Activate-Row}}
\newcommand{\mlb}{\textsc{Lower-Bound}}
\newcommand{\meb}{\textsc{Ensure-Bound-And-Get}}
\newcommand{\mblb}{\textsc{Block-Lower-Bound}}
\newcommand{\mbeb}{\textsc{Block-Ensure-Bound}}
\newcommand{\mcolmin}{\textsc{Current-Min}}
\newcommand{\mcolminrow}{\textsc{Current-Min-Row}}
\newcommand{\blcnt}{b}
\newcommand{\mblock}{B}
\newcommand{\mblocks}{\mathcal{B}}
\newcommand{\mgrcnt}{q}
\newcommand{\mgroup}{\mathcal{C}}
\newcommand{\mbnotyet}{\bar{\mblocks}}
\newcommand{\mlastg}{\texttt{last}}

\newcommand{\mfirstb}{\texttt{first}}
\newcommand{\mlastb}{\texttt{last}}
\newcommand{\mcminarr}{\texttt{cmin}}
\newcommand{\mqueue}{\textsc{Updates}}

\newcommand{\prcount}{rowcnt}
\newcommand{\pccount}{colcnt}
\newcommand{\pfcount}{rectcnt}

\newcommand{\offset}{\phi}
\newcommand{\nil}{\textbf{nil}}

\newcommand{\bpart}{block}
\newcommand{\epart}{exact}

\newcommand{\pqdec}{\textsc{Decrease-Key}}
\newcommand{\pqext}{\textsc{Extract-Min}}
\newcommand{\pqins}{\textsc{Insert}}
\newcommand{\pqminkey}{\textsc{Min-Key}}

\newcommand{\ppred}{\textsc{Pred}}
\newcommand{\psucc}{\textsc{Succ}}
\newcommand{\off}{d}
\newcommand{\om}{\text{off}}

\begin{abstract}
  We study the problem of computing shortest paths in so-called \emph{dense
  distance graphs}.
  Every planar graph $G$ on $n$ vertices can be partitioned into a set
    of $O(n/r)$ edge-disjoint regions
    (called an $r$-division) with $O(r)$ vertices each,
  such that each region has $O(\sqrt{r})$ vertices (called \emph{boundary}
  vertices) in common with other regions.
  A dense distance graph of a region is a complete graph containing
  all-pairs distances between its boundary nodes.
  A dense distance graph of an $r$-division is the union of the $O(n/r)$ dense
  distance graphs of the individual pieces.
  Since the introduction of dense distance graphs by \fak and Rao \cite{Fak:2006},
  computing single-source shortest paths in dense distance graphs
  has found numerous applications in fundamental planar graph algorithms.

  \fak and Rao \cite{Fak:2006} proposed an algorithm (later called \emph{FR-Dijkstra})
  for computing single-source shortest
  paths in a dense distance graph in $O\left(\frac{n}{\sqrt{r}}\log{n}\log{r}\right)$ time.
  We show an 
  $O\left(\frac{n}{\sqrt{r}}\left(\frac{\log^2{r}}{\log^2\log{r}}+\log{n}\log^{\epsilon}{r}\right)\right)$
  time algorithm for this problem,
  which is the first improvement to date over FR-Dijkstra for the important
  case when $r$ is polynomial in $n$.
  In this case, our algorithm is faster by a factor of $O(\log^2{\log{n}})$
  and implies improved upper bounds for such planar graph problems as
  multiple-source multiple-sink maximum flow, single-source all-sinks maximum flow,
  and (dynamic) exact distance oracles.
  
 \end{abstract}

\section{Introduction}
Computing shortest paths and finding maximum flows are
among the most basic graph optimization problems. Still,
even though a lot of effort has been made to construct
efficient algorithms for these problems, the known bounds
for the most general versions are not known to be tight yet.
For general digraphs with real edge lengths,
Bellman--Ford algorithm computes the shortest path tree
from a given vertex in $O(nm)$ time, where $n$ denotes
the number of vertices and $m$ is the number of edges.
This simple methods remains to be the best known
strongly-polynomial time bound, although some optimizations
in the constant factor are known~\cite{RandomizedBellmanFord}.
For the case of non-negative edge lengths, Fredman and Tarjan's \cite{Fredman:1987}
implementation of Dijkstra's algorithm achieves $O(m+n\log{n})$ time.
The maximum flow problem with real edge capacities
can be solved in $O(nm)$ time as well \cite{Orlin:2013},
but here the algorithm is much more complex.

Finding a truly subquadratic algorithm assuming $m=O(n)$
for either the single-source shortest paths or the maximum flow
seems to be very difficult.
However, the situation changes significantly
if we restrict ourselves to planar digraphs, which constitute
an important class of sparse graphs.
In this regime the goal is to obtain linear or almost linear time
complexity.
A~linear time algorithm for the single-source shortest path
problem with non-negative edge lengths was proposed
by Henzinger et al. \cite{Henzinger:1997}.
In their breakthrough paper, \fak and Rao gave the first 
nearly-linear time algorithm for the case of real edge
lengths~\cite{Fak:2006}.
Their algorithm had $O(n\log^3{n})$ time complexity.
Although their upper bounds for single-source
shortest paths were eventually improved to $O\left(n\frac{\log^2{n}}{\log\log{n}}\right)$ by
Mozes and Wulff-Nilsen \cite{Mozes:2010}, the
techniques introduced in \cite{Fak:2006} proved very useful
in obtaining not only nearly-linear time algorithms for other
static planar graph problems, but also first sublinear
dynamic algorithms for shortest paths and maximum flows.

A major contribution of \fak and Rao was introducing the \emph{dense distance graph}.
For a planar digraph $G$ partitioned into edge-disjoint regions $G_1,\ldots,G_g$,
define a boundary of a region $\bnd{G_i}$ to be the vertices of $G_i$
shared with other regions.
Let $\bnd{G}=\bigcup_i \bnd{G_i}$.
The boundary of each weakly connected component of $G_i$
can be assumed to lie on a constant number of faces of $G_i$.
A partition is called an $r$-division, if additionally $g=O(n/r)$,
$|V(G_i)|=O(r)$ and $|\bnd{G_i}|=O(\sqrt{r})$.
The dense distance graph of a region is a complete digraph on $\bnd{G_i}$
with the length of edge $(u,v)$ equal to the length of the shortest
path $u\to v$ in $G_i$.
The dense distance graph of $G$ is the union of dense distance graphs
of individual regions.
\fak and Rao showed how to compute the lengths of the shortest paths
from $s\in \bnd{G}$ to all
other vertices of $\bnd{G}$ in $O(\sum_i|\bnd{G_i}|\log{|\bnd{G_i}|}\log{|\bnd{G}|})$ time,
which is nearly-linear in the number of \emph{vertices} of this graph,
as opposed to the number of edges, i.e., $O(\sum_i|\bnd{G_i}|^2)$.
Their method is often called the \emph{FR-Dijkstra}.
For an $r$-division, FR-Dijkstra runs in $O\left(\frac{n}{\sqrt{r}}\log{n}\log{r}\right)$ time.
Based on FR-Dijkstra, they also showed how to compute the dense distance
graph itself in nearly-linear time.

Following the work of \fak and Rao, dense distance graphs and FR-Dijkstra
have become important planar graph primitives and
have been used to obtain faster algorithms for numerous problems
related to cuts (e.g. \cite{Borradaile:2015, Borradaile:2016, Italiano:2011}), flows (\cite{Borradaile:2011, Lacki:2012})
and computing exact point-to-point distances (\cite{Fak:2006, Mozes:2012}) in planar digraphs.
FR-Dijkstra has also found applications in algorithms for
bounded-genus graphs, e.g., \cite{Borradaile:2016}.

Although better algorithms (running in $O(\sum_i|\bnd{G_i}|\log{|\bnd{G_i}}|)$ time)
have been proposed for computing the dense distance graph itself
(\cite{Italiano:2011, Klein:2005}), the only improvement over FR-Dijkstra
to date is due to Mozes et al. (\cite{Mozes:2014}, manuscript).
Using the methods of \cite{Henzinger:1997}, they show that for an $r$-division,
the shortest paths in a dense distance graph can be found in $O\left(\frac{n}{\sqrt{r}}\log^2{r}\right)$
time.
However, this does not improve over FR-Dijkstra in the
case when $r$ is polynomial in $n$, a case which emerges in many applications.

In this paper we show an algorithm for computing single-source shortest paths
in a dense distance graph in 
$O\left(\sum_i|\bnd(G_i)|\left(\frac{\log^2{|\bnd{G_i}|}}{\log^2\log{|\bnd{G_i}|}}+\
\log{|\bnd{G}|}\log^{\epsilon}{|\bnd{G_i}|}\right)\right)$ time (for any $\epsilon\in(0,1)$),
which is faster than FR-Dijkstra in \emph{all} cases.
Specifically, in the case of an $r$-division with $r=\text{poly}(n)$, the algorithm runs in
$O\left(\frac{n}{\sqrt{r}}\frac{\log^2{n}}{\log^2\log{n}}\right)$
time.
Our algorithm implies an improvement by a factor of $O(\log^2\log{n})$ in the time complexity
for a number of planar
digraph problems such as
multiple-source multiple-sink maximum flows, maximum bipartite matching~\cite{Borradaile:2011},
single-source all-sinks maximum flows~\cite{Lacki:2012},
exact distance oracles~\cite{Mozes:2012},
It also yields polylog-logarithmic improvements to
dynamic algorithms for both
shortest paths and maximum flows~\cite{Italiano:2011, Kaplan:2012, Klein:2005}.

However, for small values of $r$, such as $r=\text{polylog}(n)$, our algorithm does
not improve on \cite{Mozes:2014}, as the $O\left(\frac{n}{\sqrt{r}}\log{n}\log^\epsilon{r}\right)$
term starts to dominate the overall complexity of our algorithm.
Dense distance graphs for $r$-divisions with $r=\text{polylog}(n)$
have also found applications, most notably in the $O(n\log\log{n})$ algorithm
for minimum $s,t$-cut in undirected planar graphs \cite{Italiano:2011}.
However, computing shortest paths in a DDG is not a bottleneck in this case.
For other applications of $r$-divisions with small $r$, consult \cite{Mozes:2014}.

\paragraph{Overview of the Result}
In order to obtain the speedup we use a subtle combination of techniques.
The problem of computing the single-source shortest paths
in a dense distance graphs is solved with an optimized
implementation of Dijkstra's algorithm.
Since the vertices of $\bnd{G_i}$ lie on $O(1)$ faces
of a planar digraph $G_i$, we can exploit the fact that
many of the shortest paths represented by the dense distance
graph have to cross.
Consequently, there is no point in relaxing most of the edges
of the dense distance graph of $G_i$.
The edge-length matrix of a dense distance graph on $G_i$
can be partitioned into a constant number of so-called
\emph{staircase Monge matrices}.
A natural approach to restricting the number of edges of the dense
distance graph to be relaxed is to design a data structure
reporting the column minima of a certain staircase Monge matrix $\mmat$
in an online fashion.
Specifically, the data structure has to handle
\emph{row activations} intermixed with extractions of the column minima
in non-decreasing order.
Once Dijkstra's algorithm establishes the distance $d(v)$
to some vertex $v$, the row of~$\mmat$ corresponding to $v$
is activated and becomes available to the data structure.
This row contains values $d(v)+\ell(v,w)$, where $\ell(v,w)$ is the length
of the edge $(v,w)$ of the DDG.
Alternatively, a minimum in some column corresponding to $v$
(in the revealed part of $\mmat$)
may be used by Dijkstra's
algorithm to establish a new distance label $d(v)$, even though
not all rows of $\mmat$ have been revealed so far.
In this case, we can guarantee that all the inactive rows of $\mmat$
contain entries not smaller than $d(v)$ and hence we can safely extract
the column minimum of $\mmat$.

We show how to use such a data structure to obtain an improved
single-source shortest path algorithm in Section~\ref{s:dijkstra}.
Such an approach was also used by \fak and Rao \cite{Fak:2006} and
Mozes et al. \cite{Mozes:2014},
who both dealt with staircase Monge matrices by using a recursive
partition into square Monge matrices, which are easier to handle.
In particular, \fak and Rao showed that a sequence of row activations
and column minima extractions can be performed on a $\msize\times\msize$
square Monge matrix in $O(\msize\log{\msize})$ time.
The recursive partition assigns each row and column to $O(\log{|\bnd{G_i}|})$ square
Monge matrices.
As a result, the total time for handling all the square matrices is
$O(|\bnd{G_i}|\log^2{|\bnd{G_i}|})$.

Our first component is a refined data structure for handling row activations
and column minima extractions on a \emph{rectangular} Monge matrix,
described in Section~\ref{s:full_min}.
We show a data structure supporting any sequence of operations on a $\mrowcnt\times\mcolcnt$ matrix in
$O\left(\mrowcnt\frac{\log{\msize}}{\log\log{\msize}}+\mcolcnt\log{\msize}\right)$ total time,
where $\msize=\max(\mrowcnt,\mcolcnt)$.
In comparison to \cite{Fak:2006}, we do not map all the columns
to active rows containing the current minima.
Instead, the columns are assigned \emph{potential row sets} of bounded size
that are guaranteed to contain the ``currently optimal'' rows.
This relaxed notion allows to remove the seemingly unavoidable binary search
at the heart of \cite{Fak:2006} and instead use the SMAWK algorithm \cite{SMAWK:1987}
to split the potential row sets once they become too large.
The maintenance of a priority queue used for reporting the column minima in order
is possible with
the recent efficient data structure supporting subrow minimum queries in Monge matrices \cite{Gawry:2015}
and the usage of priority queues with $O(1)$ time
$\pqdec$ operation \cite{Fredman:1987}.

The second step is to relax the requirements posed on a data structure
handling rectangular $\mrowcnt\times\mcolcnt$ Monge matrices.
It is motivated by the following observation.
Let $\Delta>0$ be an integer.
Imagine we have found the minima of $\mcolcnt/\Delta$ evenly spread, \emph{pivot} columns
$\mcol_1,\ldots,\mcol_{\mcolcnt/\Delta}$.
Denote by $\mrow_1,\ldots,\mrow_{\mcolcnt/\Delta}$ the rows containing
the corresponding minima.
A well-known property of Monge matrices implies that
for any column $\mcol'$ lying between $\mcol_i$ and $\mcol_{i+1}$,
we only have to look for a minimum of $\mcol'$ in rows $\mrow_i,\ldots,\mrow_{i+1}$.
Thus, the minima in the remaining columns can be found in $O(\mrowcnt\Delta+\mcolcnt)$
total time.
In Section~\ref{s:block_min} we show how to adapt this idea to an online setting
that fits our needs.
The columns are partitioned into $O(\mcolcnt/\Delta)$ \emph{blocks} of size at most $\Delta$.
Each block is conceptually contracted to a single column: an entry in row $\mrow$
is defined as the minimum in row $\mrow$ over the contracted columns.
For sufficiently small values of $\Delta$, such a minimum can be computed
in $O(1)$ time using the data structure of \cite{Gawry:2015}.
Locating a block minimum can be seen as an introduction of a new pivot column.
We handle the block matrix with the data structure of Section~\ref{s:full_min}
and prove that the total time needed to correctly report all the column
minima is $O\left(\mrowcnt\frac{\log{\msize}}{\log\log{\msize}}+\mrowcnt\Delta+\mcolcnt+\frac{\mcolcnt}{\Delta}\log{\msize}\right)$.
In particular, for $\Delta=\log^{1-\epsilon}{\msize}$, this bound becomes
$O\left(\mrowcnt\frac{\log{\msize}}{\log\log{\msize}}+\mcolcnt\log^{\epsilon}{\msize}\right)$.

Finally, in Section~\ref{s:stair_min} we exploit the asymmetry of per-row and per-column costs
of the developed block data structure for rectangular matrices
by using a different partition of a staircase Monge
matrix.
Our partition is biased towards columns, i.e., the matrix is split into
\emph{rectangular} (as opposed to square) Monge matrices, each with roughly poly-logarithmically more columns
than rows.
Consequently, the total number of rows in these matrices is $O\left(|\bnd{G_i}|\frac{\log{|\bnd{G_i}|}}{\log\log{|\bnd{G_i}|}}\right)$,
whereas the total number of columns is only slightly larger, i.e.,
$O\left(|\bnd{G_i}|\log^{1+\epsilon}{|\bnd{G_i}|}\right)$.
This yields a data structure handling staircase Monge matrices in
$O\left(|\bnd{G_i}|\frac{\log^2{|\bnd{G_i}|}}{\log^2\log{|\bnd{G_i}|}}\right)$
total time.

\paragraph{Model of Computation} We assume the standard word-RAM model with word size $\Omega(\log{n})$.
However, we stress that our algorithm works in the very general case of
\emph{real} edge lengths, i.e., we are only allowed to perform arithmetical
operations on lengths and compare them.

\paragraph{Outline of the Paper}
We present our algorithm in a bottom-up manner:
in Section~\ref{s:pre} we introduce the terminology, while
in Sections~\ref{s:full_min}, \ref{s:block_min} and~\ref{s:stair_min} 
we develop the increasingly more powerful data structures
for reporting column minima in online Monge matrices.
Each of these data structures is used in a black-box
manner in the following section.
The improved algorithm for computing single-source shortest paths
in dense distance graph is
discussed in detail in Section~\ref{s:dijkstra}.
We describe the most important implications in Section~\ref{s:impli}.

\section{Preliminaries}\label{s:pre}
\subsection{Partitions of Planar Graphs and Dense Distance Graphs}\label{s:prelim}
Let $G=(V,E)$ be a planar weighted digraph.
Let $E_1,\ldots,E_g$ be a partition of $E$ into non-empty, disjoint subsets.
We define \emph{regions} of $G$ to be the induced subgraphs $G_i=G[E_i]$.
The \emph{boundary} $\bnd{G_i}$ of a region $G_i$ is defined to be the set
of vertices of $G_i$ that also belong to other regions, i.e.,
$\bnd{G_i}=V(G_i)\cap V(G[E\setminus E_i])$.
Let $\bnd{G}=\bigcup_{i=1}^g \bnd{G_i}$.

A partition of $G$ into regions $G_1,\ldots,G_g$ is called a \emph{partition
with few holes} if for each $i$ and for each weakly connected component $G_i^j$ of $G_i$, the
vertices of $\bnd{G_i}$ lie on $O(1)$ faces of $G_i^j$.
A \emph{hole} is thus defined to be a face of $G_i$ containing
at least one vertex of $\bnd{G_i}$.

For a partition of $G$ with few holes,
we denote by $\ddg(G_i)$ the \emph{dense distance graph} of a region~$G_i$,
which is defined to be a
complete directed graph on vertices $\bnd{G_i}$, such that the weight
of an edge $(u,v)$ is equal to the length of the shortest path $u\to v$
in $G_i$.
If the edge lengths are non-negative, the dense distance graphs are typically computed using the \emph{multiple-source
shortest paths} data structure of Klein \cite{Klein:2005}.
This data structure allows us to preprocess a plane graph $G=(V,E)$ with a distinguished
face $F$ in $O(n\log{n})$ time so that we can find in $O(\log{n})$ time
the length of the shortest path $u\to v$ for any $u\in F$ and $v\in V$.
As the boundary vertices in each component of $G_i$ lie on $O(1)$
faces, we can compute $\ddg(G_i)$ in $O((|V(G_i)|+|\bnd{G_i}|^2)\log{|V(G_i)|})$ time.
$\ddg(G)$ is defined as $\bigcup_{i=1}^g\ddg(G_i)$.

For $r<n$, an $r$-division of a planar graph $G$ is a partition $G_1,\ldots,G_g$
with few holes such that $g=O(n/r)$ while $|V(G_i)|=O(r)$ and $|\bnd{G_i}|=O(\sqrt{r})$ for any $i=1,\ldots,g$.
Klein et al. \cite{Klein:2013} proved that for any triangulated and biconnected planar graph $G$
and any $r<n$, an $r$-division can be computed in linear time.
Given an $r$-division $G_1,\ldots,G_g$, $\ddg(G)$ can be thus computed
by computing the dense distance graph
for each region separately in $O(n\log{r})$ total time.

\subsection{Matrices and Their Minima}
In this paper we define a \emph{matrix} to be a partial function
$\mmat:\mrows\times\mcols\to\mathbb{R}$, where $\mrows$ (called \emph{rows})
and~$\mcols$ (called \emph{columns}) are some totally ordered finite sets.
Set $\mrows=\{\mrow_1,\ldots,\mrow_\mrowcnt\}$ and $\mcols=\{\mcol_1,\ldots,\mcol_\mcolcnt\}$,
where $\mrow_1\leq\ldots\leq\mrow_\mrowcnt$
and $\mcol_1\leq\ldots\leq\mcol_\mcolcnt$.
If for $\mrow_i,\mrow_j\in\mrows$ we have $\mrow_i\leq \mrow_j$, we also say
that $\mrow_i$ is (weakly) \emph{above} $\mrow_j$
and $\mrow_j$ is (weakly) \emph{below} $\mrow_i$.
Similarly, when $\mcol_i,\mcol_j$ we have $\mcol_i<\mcol_j$, we
say that $\mcol_i$ is \emph{to the left} of $\mcol_j$
and $\mcol_j$ is \emph{to the right} of $\mcol_i$.

For some matrix $\mmat$ defined on rows $\mrows$ and columns $\mcols$,
for $\mrow\in\mrows$ and $\mcol\in\mcols$ we denote by
$\mmat_{\mrow,\mcol}$ an \emph{element} of $\mmat$.
An element is the value of $\mmat$ on pair $(\mrow,\mcol)$,
if defined.

For $\mrows'\subseteq\mrows$ and $\mcols'\subseteq\mcols$
we define $\mmat(\mrows',\mcols')$ to be a \emph{submatrix}
of $\mmat$.
$\mmat(\mrows',\mcols')$ is a partial function on $\mrows'\times\mcols'$
satisfying $\mmat(\mrows',\mcols')_{\mrow,\mcol}=\mmat_{\mrow,\mcol}$
for any $(\mrow,\mcol)\in\mrows'\times\mcols'$ such that
$\mmat_{\mrow,\mcol}$ is defined.

The \emph{minimum} of a matrix $\min\{\mmat\}$ is defined as the minimum
value of the partial function $\mmat$.
The \emph{column minimum} of $\mmat$ in column $\mcol$ is defined
as $\min\{\mmat(\mrows,\{\mcol\})\}$.

We call a matrix $\mmat$ \emph{rectangular} if $\mmat_{\mrow,\mcol}$
is defined for every $\mrow\in\mrows$ and $\mcol\in\mcols$.
A matrix is called \emph{staircase} (\emph{flipped staircase}) if $|\mrows|=|\mcols|$ and
$\mmat_{\mrow_i,\mcol_j}$ is defined if and only if $i\leq j$ ($i\geq j$ respectively).

Finally, a \emph{subrectangle} of $\mmat$ is a rectangular matrix
$\mmat(\{\mrow_a,\ldots,\mrow_b\}, \{\mcol_x,\ldots,\mcol_y\})$
for $1\leq a\leq b\leq \mrowcnt$, $1\leq x\leq y\leq \mcolcnt$.
We define a \emph{subrow} to be a subrectangle with a single row.

Given a matrix $\mmat$ and a function $\off:\mrows\to\mathbb{R}$,
we define the \emph{offset matrix} $\om(\mmat,\off)$ to be a matrix $\mmat'$
such that for all $\mrow\in\mrows$, $\mcol\in\mcols$ for which
$\mmat_{\mrow,\mcol}$ is defined, we have
$\mmat'_{\mrow,\mcol}=\mmat_{\mrow,\mcol}+\off(\mrow)$.

\subsection{Monge Matrices}
We say that a matrix $\mmat$ with rows $\mrows$ and columns $\mcols$
is a \emph{Monge matrix}, if for each $\mrow_1,\mrow_2\in\mrows$, $\mrow_1\leq\mrow_2$
and $\mcol_1,\mcol_2\in\mcols$, $\mcol_1\leq\mcol_2$ such
that all elements $\mmat_{\mrow_1,\mcol_1},\mmat_{\mrow_1,\mcol_2},
\mmat_{\mrow_2,\mcol_1},\mmat_{\mrow_2,\mrow_2}$ are defined,
the following \emph{Monge property} holds
$$\mmat_{\mrow_2,\mcol_1}+\mmat_{\mrow_1,\mcol_2}\leq\mmat_{\mrow_1,\mcol_1}+\mmat_{\mrow_2,\mcol_2}.$$

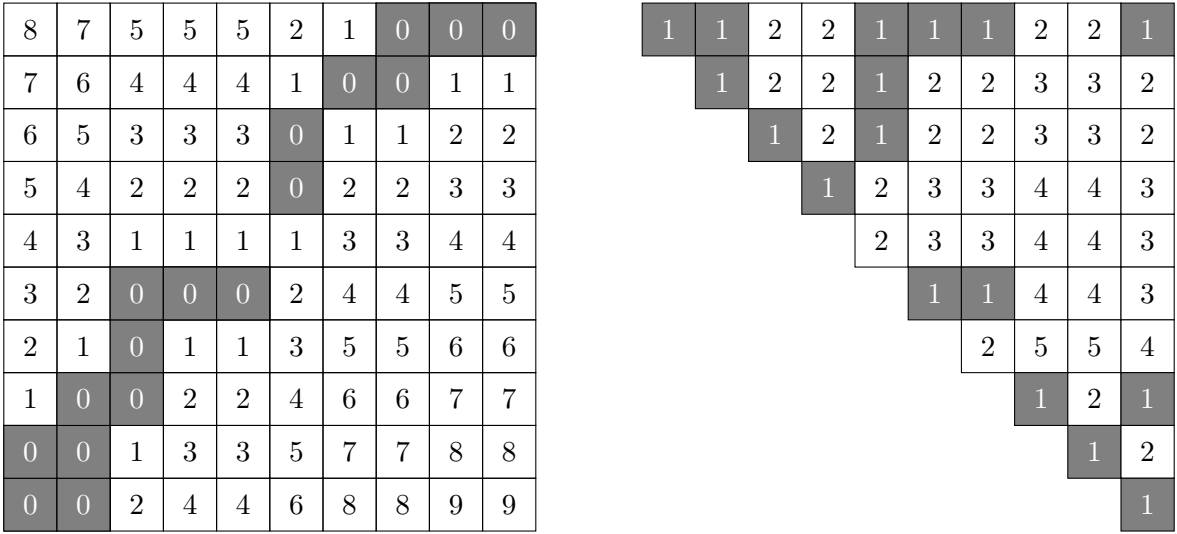
\begin{figure}[htb]
\centering
\begin{tikzpicture}[scale=0.7]

\foreach \x/\a/\b in {0/0/1,1/0/2,2/2/4,3/4/4,4/4/4,5/6/7,6/8/8,7/8/9,8/9/9,9/9/9} {
  \foreach \y in {\b,...,9} {
    \pgfmathsetmacro{\u}{\y-\b};
    \draw (\x,\y) rectangle (\x+1,\y+1) node[pos=.5,color=black] {\pgfmathprintnumber[int trunc]{\u}};
  }
  \foreach \y in {0,...,\a} {
    \pgfmathsetmacro{\u}{\a-\y};
    \draw (\x,\y) rectangle (\x+1,\y+1) node[pos=.5,color=black] {\pgfmathprintnumber[int trunc]{\u}};
  }

  \foreach \y in {\a,...,\b} {
    \draw[fill=gray] (\x,\y) rectangle (\x+1,\y+1) node[pos=.5,color=white] {0};
  }

}

\draw[fill=gray] (12,9) rectangle (13,10) node[pos=.5,color=white] {1};
\draw[fill=gray] (13,9) rectangle (14,10) node[pos=.5,color=white] {1};
\draw[fill=white] (14,9) rectangle (15,10) node[pos=.5,color=black] {2};
\draw[fill=white] (15,9) rectangle (16,10) node[pos=.5,color=black] {2};
\draw[fill=gray] (16,9) rectangle (17,10) node[pos=.5,color=white] {1};
\draw[fill=gray] (17,9) rectangle (18,10) node[pos=.5,color=white] {1};
\draw[fill=gray] (18,9) rectangle (19,10) node[pos=.5,color=white] {1};
\draw[fill=white] (19,9) rectangle (20,10) node[pos=.5,color=black] {2};
\draw[fill=white] (20,9) rectangle (21,10) node[pos=.5,color=black] {2};
\draw[fill=gray] (21,9) rectangle (22,10) node[pos=.5,color=white] {1};
\draw[fill=gray] (13,8) rectangle (14,9) node[pos=.5,color=white] {1};
\draw[fill=white] (14,8) rectangle (15,9) node[pos=.5,color=black] {2};
\draw[fill=white] (15,8) rectangle (16,9) node[pos=.5,color=black] {2};
\draw[fill=gray] (16,8) rectangle (17,9) node[pos=.5,color=white] {1};
\draw[fill=white] (17,8) rectangle (18,9) node[pos=.5,color=black] {2};
\draw[fill=white] (18,8) rectangle (19,9) node[pos=.5,color=black] {2};
\draw[fill=white] (19,8) rectangle (20,9) node[pos=.5,color=black] {3};
\draw[fill=white] (20,8) rectangle (21,9) node[pos=.5,color=black] {3};
\draw[fill=white] (21,8) rectangle (22,9) node[pos=.5,color=black] {2};
\draw[fill=gray] (14,7) rectangle (15,8) node[pos=.5,color=white] {1};
\draw[fill=white] (15,7) rectangle (16,8) node[pos=.5,color=black] {2};
\draw[fill=gray] (16,7) rectangle (17,8) node[pos=.5,color=white] {1};
\draw[fill=white] (17,7) rectangle (18,8) node[pos=.5,color=black] {2};
\draw[fill=white] (18,7) rectangle (19,8) node[pos=.5,color=black] {2};
\draw[fill=white] (19,7) rectangle (20,8) node[pos=.5,color=black] {3};
\draw[fill=white] (20,7) rectangle (21,8) node[pos=.5,color=black] {3};
\draw[fill=white] (21,7) rectangle (22,8) node[pos=.5,color=black] {2};
\draw[fill=gray] (15,6) rectangle (16,7) node[pos=.5,color=white] {1};
\draw[fill=white] (16,6) rectangle (17,7) node[pos=.5,color=black] {2};
\draw[fill=white] (17,6) rectangle (18,7) node[pos=.5,color=black] {3};
\draw[fill=white] (18,6) rectangle (19,7) node[pos=.5,color=black] {3};
\draw[fill=white] (19,6) rectangle (20,7) node[pos=.5,color=black] {4};
\draw[fill=white] (20,6) rectangle (21,7) node[pos=.5,color=black] {4};
\draw[fill=white] (21,6) rectangle (22,7) node[pos=.5,color=black] {3};
\draw[fill=white] (16,5) rectangle (17,6) node[pos=.5,color=black] {2};
\draw[fill=white] (17,5) rectangle (18,6) node[pos=.5,color=black] {3};
\draw[fill=white] (18,5) rectangle (19,6) node[pos=.5,color=black] {3};
\draw[fill=white] (19,5) rectangle (20,6) node[pos=.5,color=black] {4};
\draw[fill=white] (20,5) rectangle (21,6) node[pos=.5,color=black] {4};
\draw[fill=white] (21,5) rectangle (22,6) node[pos=.5,color=black] {3};
\draw[fill=gray] (17,4) rectangle (18,5) node[pos=.5,color=white] {1};
\draw[fill=gray] (18,4) rectangle (19,5) node[pos=.5,color=white] {1};
\draw[fill=white] (19,4) rectangle (20,5) node[pos=.5,color=black] {4};
\draw[fill=white] (20,4) rectangle (21,5) node[pos=.5,color=black] {4};
\draw[fill=white] (21,4) rectangle (22,5) node[pos=.5,color=black] {3};
\draw[fill=white] (18,3) rectangle (19,4) node[pos=.5,color=black] {2};
\draw[fill=white] (19,3) rectangle (20,4) node[pos=.5,color=black] {5};
\draw[fill=white] (20,3) rectangle (21,4) node[pos=.5,color=black] {5};
\draw[fill=white] (21,3) rectangle (22,4) node[pos=.5,color=black] {4};
\draw[fill=gray] (19,2) rectangle (20,3) node[pos=.5,color=white] {1};
\draw[fill=white] (20,2) rectangle (21,3) node[pos=.5,color=black] {2};
\draw[fill=gray] (21,2) rectangle (22,3) node[pos=.5,color=white] {1};
\draw[fill=gray] (20,1) rectangle (21,2) node[pos=.5,color=white] {1};
\draw[fill=white] (21,1) rectangle (22,2) node[pos=.5,color=black] {2};
\draw[fill=gray] (21,0) rectangle (22,1) node[pos=.5,color=white] {1};
\end{tikzpicture}
\caption{Example $10\times 10$ Monge matrices: a rectangular one to the left and a staircase
one to the right. The grey cells contain the column minima of the respective columns.}
\label{f:monge}
\end{figure}

\begin{fact}\label{f:submatrix}
Let $\mmat$ be a Monge matrix.
For any $\mrows'\subseteq\mrows$ and $\mcols'\subseteq\mcols$,
$\mmat(\mrows',\mcols')$ is also a Monge matrix.
\end{fact}

\begin{fact}\label{f:condense}
Let $\mmat$ be a rectangular Monge matrix and assume $\mrows$
is partitioned into disjoint blocks $\mathcal{R}=\mrows_1,\ldots,\mrows_a$
such that each $\mrows_i$ is a contiguous group of subsequent rows
and each~$\mrows_i$ is above~$\mrows_{i+1}$.
Assume also that the set $\mcols$ is partitioned into
blocks $\mathcal{C}=\mcols_1,\ldots,\mcols_b$ so that $\mcols_i$ is to the left
of $\mcols_{i+1}$.
Then, a matrix $\mmat'$ with rows $\mathcal{R}$ and columns $\mathcal{C}$
defined as
$$\mmat'_{\mrows_i,\mcols_j}=\min\{\mmat(\mrows_i,\mcols_j)\},$$
is also a Monge matrix.
\end{fact}

\begin{fact}\label{f:monoton}
Let $\mmat$ be a rectangular Monge matrix. Assume that for some
$\mcol\in\mcols$ and $\mrow\in\mrows$, $\mmat_{\mrow,\mcol}$ is
a column minimum of $\mcol$.
Then, for each column $\mcol^-$ to the left of $\mcol$,
there exists a row $\mrow^-$ (weakly) below $\mrow$,
such that $\mmat_{\mrow^-,\mcol^-}$ is a column minimum
of $\mcol^-$.
Similarly, for each column $\mcol^+$ to the right of $\mcol$,
there exists a row $\mrow^+$ (weakly) above $\mrow$,
such that $\mmat_{\mrow^+,\mcol^+}$ is a column minimum
of $\mcol^+$.
\end{fact}

\begin{fact}\label{f:contig}
Let $\mmat$ be a rectangular Monge matrix. Let $\mrow\in\mrows$
and $\mcols=\{\mcol_1,\ldots,\mcol_\mcolcnt\}$.
The set of columns $\mcols_\mrow\in\mcols$ having
one of their column minima in row $\mrow$ is 
contiguous, that is either $\mcols_\mrow=\emptyset$
or $\mcols_\mrow=\{\mcol_a,\ldots,\mcol_b\}$
for some $1\leq a\leq b\leq \mcolcnt$.
\end{fact}

\begin{remark}
  The statements of facts~\ref{f:monoton}~and~\ref{f:contig} could be simplified
  if we either assumed that the column minima in the considered Monge matrices
  are unique or introduced some tie-breaking rule.
  However, this would lead to a number of similar assumptions in terms of priority
  queue keys and path lengths in the following sections which in turn would
  complicate the description.
  Thus, we do not use any simplifying assumptions about the column minima.
\end{remark}

\begin{fact}\label{f:offset}
  Let $\mmat$ be a Monge matrix with rows $\mrows$ and let $\off:\mrows\to\mathbb{R}$.
  Then $\om(\mmat,\off)$ is also a Monge matrix.
\end{fact}

\subsection{Data-Structural Prerequisites}

\paragraph{Priority Queues}
We assume that priority queues store elements with real keys.
A priority queue $H$ supports the following set of operations:
\begin{itemize}
\item $\pqins(e,k)$ -- insert an element $e$ with key $k$ into $H$.
\item $\pqext()$ -- delete an element $e\in H$ with the smallest key and return $e$.
\item $\pqdec(e,k)$ -- given an element $e\in H$, decrease key of $e$ to $k$.
  If the current key of $e$ is smaller than $k$, do nothing.
\item $\pqminkey()$ -- return the smallest key in $H$.
\end{itemize}
Formally, we assume that each call $\pqins(e,k)$ also produces a ``handle'',
which can be later used to point the call $\pqdec$ to a place inside
$H$, where $e$ is being kept.
In our applications, the elements stored in a priority queue are always
distinct and thus for brevity we skip the details of using handles later on.

Fredman and Tarjan \cite{Fredman:1987} showed a data structure called \emph{the Fibonacci heap},
which can perform \linebreak
$\pqext$ in amortized $O(\log{n})$ time and all the remaining
operations in amortized $O(1)$ time.
Here $n$ is the current size of the queue.
In the following sections, we assume that each priority queue is implemented
as a Fibonacci heap.

\paragraph{Predecessor Searching}
Let $S$ be some totally ordered set such that for any $s\in S$ we can compute
the \emph{rank} of $s$, i.e., the number $|\{y\leq s:y\in S\}|$, in constant time.
A \emph{dynamic predecessor/successor data structure} maintains a subset $R$ of $S$
and supports the following operations:
\begin{itemize}
\item Insertion of some $s\in S$ into $R$.
\item Deletion of some $s\in R$.
\item $\ppred(s)$ ($\psucc(s)$) -- for some $s\in S$, return the largest
  (smallest respectively) element $r$ of $R$ such that $r\leq s$ ($r\geq s$ resp.).
\end{itemize}
Van Emde Boas \cite{vEB:1977} showed that using $O(|S|)$ space
we can perform each of these operations in $O(\log\log{|S|})$
time.
Whenever we use a dynamic predecessor/successor data structure
in the following sections,
we assume the above bounds to hold.

\section{Online Column Minima of a Rectangular Offset Monge Matrix}\label{s:full_min}
Let $\mmat_0$ be a rectangular $\mrowcnt\times\mcolcnt$ Monge matrix.
Let $R=\{\mrow_1,\ldots,\mrow_\mrowcnt\}$ and
$C=\{\mcol_1,\ldots,\mcol_\mcolcnt\}$ be the sets of rows
and columns of $\mmat_0$, respectively.
Set $\msize=\max(\mrowcnt,\mcolcnt)$.

Let $\off:\mrows\to\mathbb{R}$ be an offset function and set $\mmat=\om(\mmat_0,\off)$.
By Fact~\ref{f:offset}, $\mmat$ is also a Monge matrix.
Our goal is to design a data structure capable of reporting
the column minima of $\mmat$ in increasing order of their values.
However, the function $\off$ is not entirely revealed beforehand,
as opposed to the matrix $\mmat_0$.
There is an initially empty, growing set $\mcurrows\subseteq\mrows$
containing the rows for which $\off(\mrow)$ is known.
Alternatively, $\mcurrows$ can be seen as a set of ``active''
rows of $\mmat$ which can be accessed by the data structure.
There is also a set $\mcurcols\subseteq\mcols$ containing the remaining
columns for which we have not reported the minima yet.
Initially, $\mcurcols=\mcols$ and the set $\mcurcols$
shrinks over time.
We also provide a mechanism to guarantee that the rows
that have not been revealed do not influence the smallest
of the column minima of $\mcurcols$.

The exact set of operations we support is the following:
\begin{itemize}
\item $\minit(\mrows,\mcols)$ -- initialize the data structure and set
  $\mcurrows=\emptyset$, $\mcurcols=\mcols$.
\item $\marow(\mrow)$, where $\mrow\in\mrows\setminus\mcurrows$ --
  add $\mrow$ to the set $\mcurrows$.
\item $\mlb()$ -- compute the number $\min\{\mmat(\mcurrows,\mcurcols)\}$.
  If $\mcurrows=\emptyset$ or $\mcurcols=\emptyset$, return $\infty$.
\item $\meb()$ -- inform the data structure that we have
  $$\min\{\mmat(\mrows\setminus\mcurrows,\mcols)\}\geq\min\{\mmat(\mcurrows,\mcurcols)\}=\mlb(),$$
  that is, the smallest element of $\mmat(\mrows,\mcurcols)$ does not
  depend on the values of $\mmat$ located in rows $\mrows\setminus\mcurrows$.
  It is the responsibility of the user to guarantee that this condition
  is in fact satisfied.

  Such claim implies that for some column $\mcol\in\mcurcols$ we have
  $\min\{\mmat(\mrows,\{\mcol\})\}=\min\{\mmat(\mcurrows,\mcurcols)\}$,
  which in turn means that we are able to find the minimum element in column $c$.
  The function returns any such $c$ and removes it from the set $\mcurcols$.

\item $\mcolminrow(\mcol)$, where $\mcol\in\mcols$ -- compute $\mrow$, where $\mrow\in\mcurrows$
  is a row
  such that $\min\{\mmat(\mcurrows,\{\mcol\})\}=\mmat_{\mrow,\mcol}$.
  If $\mcurrows=\emptyset$, return $\nil$.
  Note that $\mcol$ is not necessarily in~$\mcurcols$.
  
  Additionally, we require $\mcolminrow$ to have the following property:
  once the column $\mcol$ is moved out of $\mcurcols$, $\mcolminrow(\mcol)$ always returns
  the same row.
  Moreover, for $\mcol_1,\mcol_2\in\mcols$ such that $\mcol_1<\mcol_2$
  we require $$\mcolminrow(\mcol_1)\geq \mcolminrow(\mcol_2).$$
\end{itemize}
Note that $\marow$ increases the size of $\mcurrows$ and thus cannot
be called more than~$\mrowcnt$ times.
Analogously, $\meb$ decreases the size of $\mcurcols$ so it cannot be called
more than $\mcolcnt$ times.
Actually, in order to reveal all the column minima with this data structure,
the operation $\meb$ has to be called \emph{exactly} $\mcolcnt$ times.

\subsection{The Components}

\paragraph{The Subrow Minimum Query Data Structure}
Given $\mrow\in\mcurrows$ and $a,b$, $1\leq a\leq b\leq\mcolcnt$,
a subrow minimum query $S(\mrow,a,b)$ computes a column
$\mcol\in\{\mcol_a,\ldots,\mcol_b\}$ such that
$$\min\{\mmat(\{\mrow\},\{\mcol_a,\ldots,\mcol_b\})\}=\mmat_{\mrow,\mcol}.$$
We use the following theorem of Gawrychowski et al. \cite{Gawry:2015}.

\begin{theorem}[Theorem 3. of \cite{Gawry:2015}]\label{t:subrow}
Given a $\mrowcnt\times\mcolcnt$ rectangular Monge matrix $\mmat$, a data structure
of size $O(\mcolcnt)$ can be constructed in $O(\mcolcnt\log{\mrowcnt})$ time
to answer subrow minimum queries in $O(\log{\log{(\mrowcnt+\mcolcnt)}})$ time.
\end{theorem}

Recall that $\mmat=\om(\mmat_0,\off)$.
Adding the offset $\off(\mrow)$ to all the elements in row $\mrow$ of
$\mmat_0$ does not change the relative order of elements in row $r$.
Hence, the answer to a subrow minimum query $S(\mrow,a,b)$ in $\mmat$ is 
the same as the answer to $S(\mrow,a,b)$ in $\mmat_0$.

We build a data structure of Theorem~\ref{t:subrow} for $\mmat_0$
and assume that any subrow minimum query in $\mmat$ can be answered
in $O(\log\log{m})$ time.

\paragraph{The Column Groups}
The set $\mcols$ is internally partitioned into disjoint, contiguous \emph{column groups} 
$\mgroup_1,\ldots,\mgroup_\mgrcnt$ (where~$\mgroup_1$ is the leftmost group
and $\mgroup_\mgrcnt$ is the rightmost),
so that $\bigcup_i\mgroup_i=\mcols$.

As the groups constitute contiguous segments of columns, we can represent
the partition with a subset $F\subseteq \mcols$ containing the first columns
of individual groups.
Each group can be identified with its leftmost column.
We use a dynamic predecessor data structure for maintaining the set $F$.
The first column of the group containing column $\mcol$ can be
thus found by calling $F.\ppred(\mcol)$ in $O(\log\log{m})$ time.
Such representation also allows to split groups and merge neighboring groups
in $O(\log\log{m})$ time.

\paragraph{The Potential Row Sets}
For each $\mgroup_i$ we store a set $P(\mgroup_i)\subseteq \mcurrows$,
called a \emph{potential row set}.
Between consecutive operations, the potential row sets satisfy
the following invariants:
\begin{enumerate}[label=P.\arabic*,leftmargin=*]
\item For any $\mcol\in\mgroup_i$ there exists a row $\mrow\in P(\mgroup_i)$
  such that $\min\{\mmat(\mcurrows,\{\mcol\})\}=\mmat_{\mrow,\mcol}$.
  \label{i:prow}
\item The size of any set $P(\mgroup_i)$ is less than $2\alpha$, where $\alpha$ is a parameter
  to be fixed later.\label{i:psize}
\item For any $i<j$ and any $\mrow_i\in P(\mgroup_i)$, $\mrow_j\in P(\mgroup_j)$,
  we have $\mrow_i\geq \mrow_j$.\label{i:pmono}
\end{enumerate}

As by Fact~\ref{f:submatrix}
$\mmat(\mcurrows,\mcols)$ is a Monge matrix,
from Fact~\ref{f:monoton} it follows that invariant~\ref{i:pmono} can
be indeed satisfied.
By invariant~\ref{i:pmono} we also have $|P(\mgroup_i)\cap P(\mgroup_{i+1})|\leq 1$
and thus the sum of sizes of sets $P(\mgroup_i)$ is $O(\mrowcnt+\mcolcnt)$.
The sets $P(\mgroup_i)$ are stored as balanced binary search trees,
sorted bottom to top.
Additionally, the union of sets~$P(\mgroup_i)$ is
stored in a dynamic predecessor/successor data structure $U$.
We also have an auxiliary array $\mlastg$ mapping each row $\mrow\in\mcurrows$
to the rightmost column group $\mgroup_i$ such that $\mrow\in P(\mgroup_i)$
(if such group exists).
\begin{lemma}\label{l:pcost}
An insertion or deletion of some $\mrow$ to $P(\mgroup_i)$ (along with the update
of the auxiliary structures) can be performed in $O(\log{\alpha}+\log{\log{m}})$ time.
\end{lemma}
\begin{proof}
  The cost of updating the binary search tree is $O(\log{|P(\mgroup_i)|})=O(\log{\alpha})$, whereas
updating the predecessor structure $U$ takes $O(\log\log{m})$ time.
Updating the array $\mlastg$ upon insertion is trivial.
When a row $r$ is deleted and $\mlastg[r]\neq\mgroup_i$, $\mlastg[r]$
does not have to be updates.
Otherwise, we check if $r\in P(\mgroup_{i-1})$ and set $\mlastg[r]$
to either $\mgroup_{i-1}$ or $\nil$.
\end{proof}

\paragraph{Special Handling of Columns with Known Minima}
We require that for each column $\mcol$ being moved out of $\mcurcols$,
a row $y_\mcol$ such that $\min\{\mmat(\mrows,\mcol)\}=\mmat_{y_\mcol,\mcol}$
is computed.
In order to ensure that $\mcolminrow$ has the described deterministic
behavior, we guarantee that starting at the moment of deletion
of $\mcol$ from $\mcurcols$,
there exists a group $\mgroup$ consisting
of a single element $\mcol$, such that $P(\mgroup)=\{y_\mcol\}$.
Such groups are called \emph{done}.

\paragraph{The Priority Queue}
A priority queue $H$ contains an element $\mcol$
for each $\mcol\in\mcurcols$.
The queue $H$ satisfies the following invariants.
\begin{enumerate}[label=H.\arabic*,leftmargin=*]
  \item For each $\mcol\in\mcurcols$, the key of $\mcol$ in $H$ is greater
    than or equal to $\min\{\mmat(\mcurrows,\{\mcol\})\}$.\label{i:hcorrect}
  \item For each group $\mgroup_j$
that is not done, there exists such column $\mcol_j\in \mgroup_j$
that the key of $\mcol_j$ in $H$ is equal to\label{i:heap}
$$\min\{\mmat(\mcurrows,\{\mcol_j\})\}=\min\left\{\mmat\left(\mcurrows,\mgroup_j\right)\right\}.$$
\end{enumerate}
\begin{lemma}\label{l:hinv}
We can ensure that invariant~\ref{i:heap} is satisfied for a single group $\mgroup_j$
in $O(\alpha\log\log{\msize})$ time.
\end{lemma}
\begin{proof}
  We perform $O(|P(\mgroup_j)|)=O(\alpha)$ subrow minimum queries on $\mmat$
  to compute for each $\mrow\in P(\mgroup_j)$
  some column $\mcol\in\mgroup_j$ such that $\mmat_{\mrow,\mcol}=\min\{\mmat(\{\mrow\},\mgroup_j)\}$.
  As each subrow minimum query takes $O(\log\log{\msize})$ time, this
  takes $O(\alpha\log\log{\msize})$ in total.
  For each computed $\mcol$, we decrease the key of $\mcol$ in $H$
  to $\mmat_{\mrow,\mcol}$ in $O(1)$ time.
  Note that by invariant~\ref{i:prow}, some $\mmat_{\mrow,\mcol}$
  is in fact equal to $\min\{\mmat(\mcurrows,\mgroup_j)\}$.
\end{proof}
We will maintain invariant~\ref{i:hcorrect} implicitly, each time setting
the key of a column $\mcol$ to either $\infty$ or some value $\mmat_{\mrow,\mcol}$,
where $\mrow\in\mcurrows$.
Note that invariant~\ref{i:heap} guarantees that the key of the top
element of $H$ is equal to $\min\{\mmat(\mcurrows,\mcurcols)\}$.

\subsection{Implementing the Operations}
\paragraph{Initialization}
First, we build the data structure of Theorem~\ref{t:subrow} in $O(\mcolcnt\log{\msize})$ time.
Then, an element $\mcol$ with key $\infty$ is inserted into $H$ for
each $\mcol\in\mcols$.
When the first row $\mrow$ is activated,
we create a single group $\mgroup=\mcols$ with $P(\mgroup)=\{\mrow\}$.
Using Lemma~\ref{l:hinv} we ensure that invariant~\ref{i:heap} is satisfied.

\paragraph{\mcolminrow} The data structure $F$ is used to identify
the group $\mgroup$ containing the column $\mcol$.
If $\mcol\in {\mcols\setminus\mcurcols}$, then the group $\mcol$ is
done and we return the only element of $P(\mgroup)$.
Otherwise, we spend $O(|P(\mgroup)|)=O(\alpha)$ time to find the topmost
row of $P(\mgroup)$ that contains a minimum of $\mcol$.
By Fact~\ref{f:monoton} and invariant~\ref{i:pmono}, returning the topmost row of $P(\mgroup)$ guarantees
that for $\mcol_1\leq\mcol_2$, $\mcolminrow(\mcol_1)\geq \mcolminrow(\mcol_2)$.
The total running time is thus $O(\alpha+\log\log{m})$.

\paragraph{\mlb, \meb}
Invariant~\ref{i:heap} guarantees that we have
$\min\{\mmat(\mcurrows,\mcurcols)\}=\min\{\mmat(\mcurrows,\{\mcol_i\})\}=\mmat_{\mrow^*,\mcol_i}=H.\pqminkey()$,
where $\mcol_i$ is the top element of $H$ and $\mrow^*$ is a row returned
by $\mcolminrow(\mcol_i)$.
Thus, the operation $\mlb()$ can be executed in $O(1)$ time.

Let us now implement $\meb$.
By the precondition of this call, we conclude that $\mmat_{\mrow^*,\mcol_i}=\min\{\mmat(\mrows,\mcurcols)\}$.
By invariant~\ref{i:heap}, $H.\pqext()$ returns the column $\mcol_i$.
With a single query to $F$, we find the current group of~$c_i$,
$\mgroup=\{\mcol_a,\ldots,\mcol_i,\ldots,\mcol_b\}$.
First, we need to create a single-column group $\mgroup^*=\{\mcol_i\}$
and mark it done, with $P(\mgroup^*)=\{\mrow^*\}$.
We thus split $\mgroup$ into at most three groups $\mgroup^-=\{\mcol_a,\ldots,\mcol_{i-1}\}$,
$\mgroup^*$ and $\mgroup^+=\{\mcol_{i+1},\ldots,\mcol_b\}$
and mark $\mgroup^*$ done.
By Fact~\ref{f:monoton}, we can safely set
$P(\mgroup^-)=\{r\in P(\mgroup):r\geq r^*\}$
and $P(\mgroup^+)=\{r\in P(\mgroup):r\leq r^*\}$.
The split of $\mgroup$ requires $O(1)$ operations on $F$,
whereas by Lemma~\ref{l:pcost}, replacing the set $P(\mgroup)$
with the sets $P(\mgroup^-),P(\mgroup^*),P(\mgroup^+)$ takes $O(\alpha(\log\log{m}+\log{\alpha}))$ time.
The last step is to fix the invariant~\ref{i:heap} for the newly created groups.
This takes $O(\alpha\log\log{\msize})$, by Lemma~\ref{l:hinv}.
Thus, taking into account the $O(\log{m})$ cost
of performing $H.\pqext$, $\meb$ takes $O(\log{m}+\alpha(\log{\alpha}+\log\log{\msize}))$ time.

Before we describe how $\marow$ is implemented, we need the following
lemma.
\begin{lemma}\label{l:split}
Let $\mmat$ be a $u\times v$ rectangular Monge matrix
with rows $\mrows=\{\mrow_1,\ldots,\mrow_u\}$ and
columns $\mcols=\{\mcol_1,\ldots,\mcol_v\}$.
For any $i\in [1,u]$,
in $O\left(u\frac{\log{v}}{\log{u}}\right)$ time we can find such
column $\mcol_s\in\mcols$ that:
\begin{enumerate}
\item Some minima of columns $\mcol_1,\ldots,\mcol_s$ lie in rows
  $\mrow_1,\ldots,\mrow_i$.
\item Some minima of columns $\mcol_{s+1},\ldots,\mcol_v$ lie in rows
  $\mrow_{i+1},\ldots,\mrow_{u}$.
\end{enumerate}
\end{lemma}
\begin{proof}

Aggarwal et al. \cite{SMAWK:1987} proved the following theorem.
The algorithm they found was nicknamed \emph{the SMAWK algorithm}.
\begin{theorem}\label{t:smawk}
One can compute the bottommost column minima of a rectangular
$\mrowcnt\times\mcolcnt$ Monge matrix in $O(\mrowcnt+\mcolcnt)$ time.
\end{theorem}

If $u\geq v$, we can find the column minima for each column of matrix $\mmat$ using
the SMAWK algorithm in $O(u)$ time.
Picking the right $\mcol_s$ is straightforward in this case.

Assume $u<v$.
We first pick a set $\mcols'=\{\mcol_1',\ldots,\mcol_u'\}$ of $u$ evenly spread columns of $\mcols$,
including the leftmost and the rightmost column.
By Fact~\ref{f:submatrix}, $\mmat(\mrows,\mcols')$ is also a Monge matrix.
The SMAWK algorithm is then used to obtain the bottommost
rows $\mrow_1',\ldots,\mrow_u'$
containing the column minima of $\mcol_1',\ldots,\mcol_u'$
in $O(u)$ time.
By Fact~\ref{f:monoton} we have $\mrow_1'\geq\ldots\geq \mrow_u'$.
We then find some $j$ such that $\mrow_j'\geq \mrow_i\geq \mrow_{j+1}'$.
The sought column~$c_s$ can now be found by proceeding
recursively on the matrix
$$\mmat'=\mmat(\mrows,\{\mcol_j',\ldots,\mcol_{j+1}'\}).$$
The matrix $\mmat'$ has still $u$ rows, but it has only $O(v/u)$ columns.

At each recursive step we divide the size of the column set by $\Omega(u)$,
so there are at most $\log_u{v}=\frac{\log{v}}{\log{u}}$ steps.
Each step takes $O(u)$ time and hence we obtain the desired bound. 
\end{proof}

\paragraph{\marow}
Assume we activate row $\mrow$.
At that point $\mrow\notin P(\mgroup_i)$ for any group $\mgroup_i$.
Our goal is to reorganize the column groups and their potential row sets
so that the conditions~\ref{i:prow}, \ref{i:psize}, \ref{i:pmono} and~\ref{i:heap}
are again satisfied.

Consider some group $\mgroup_i$. $\mgroup_i$ can fall into
three categories.
\begin{enumerate}[label=\textbf{C.\arabic*},leftmargin=*]
\item For each $\mcol\in\mgroup_i$ we have
  $\mmat_{\mrow,\mcol}\leq \min\{\mmat(P(\mgroup_i),\{\mcol\})\}$.
  \label{c1}
\item For some two columns $\mcol_1,\mcol_2\in\mgroup_i$ we have
  $\mmat_{\mrow,\mcol_1}< \min\{\mmat(P(\mgroup_i),\{\mcol_1\})\}$
and 
$\mmat_{\mrow,\mcol_2}> \min\{\mmat(P(\mgroup_i),\{\mcol_2\})\}$.\label{c2}
\item For each $\mcol\in\mgroup_i$ we have
  $\mmat_{\mrow,\mcol}\geq \min\{\mmat(P(\mgroup_i),\{\mcol\})\}$.
\label{c3}
\end{enumerate}
Fact~\ref{f:contig} guarantees that row $\mrow$ contains
column minima for a (possibly empty) interval of
columns of $\mmat(\mcurrows\cup\{\mrow\},\mcols)$.
As the groups do not overlap, this implies that the groups in
category~\ref{c1} form a (possibly empty) interval of groups $\mgroup_a,\ldots,\mgroup_b$,
while there are at most two category~\ref{c2} groups -- $\mgroup_{a-1}$
and $\mgroup_{b+1}$, if they exist.
The groups that are done, clearly fall into category~\ref{c3}.

We can decide if $\mgroup_i$ falls into category~\ref{c1} 
in $O(|P(\mgroup_i)|)=O(\alpha)$ time by
looking only at the leftmost and rightmost columns $\mcol_-, \mcol_+$
of~$\mgroup_i$.
Clearly, if for some $\mrow'\in P(\mgroup_i)$ we have
$\mmat_{\mrow',\mcol_-}<\mmat_{\mrow,\mcol_-}$ or
$\mmat_{\mrow',\mcol_+}<\mmat_{\mrow,\mcol_+}$, $\mgroup_i$
does not belong to~\ref{c1}.
Otherwise, by invariant~\ref{i:prow}, the row $\mrow$ contains some column minima
of both columns $\mcol_-$ and $\mcol_+$ of $\mmat(\mcurrows\cup\{\mrow\},\mgroup_i)$
and hence by Fact~\ref{f:contig} it contains
column minima for all columns of $\mgroup_i$.
Moreover, if $\mrow$ is below all the rows of $P(\mgroup_i)$
or above all the rows of $P(\mgroup_i)$,
by looking only at the border columns of $\mgroup_i$, we can precisely
detect the category of $\mgroup_i$.
As invariant~\ref{i:pmono} holds before the activation of $\mrow$,
there is at most one group $\mgroup^{+}_{-}$ such that $P(\mgroup^{+}_{-})$
contains rows both above and below $\mrow$.

We first find the
rightmost group $\mgroup_i$ such that for all $\mrow'\in P(\mgroup_i)$
we have $\mrow'> \mrow$.
This can be done in $O(\log\log{m})$ time by setting
$\mgroup_i=\mlastg(U.\psucc(\mrow))$.
By Fact~\ref{f:contig}, if there
is any group $\mgroup'$ in categories~\ref{c1} or~\ref{c2},
then one of the groups $\mgroup_i$, $\mgroup_{i+1}$
also falls into~\ref{c1} or~\ref{c2}.
We may thus find all groups $\mgroup_a,\ldots,\mgroup_b$
in category~\ref{c1}
by moving both to the left and to the right of~$\mgroup_i$.
The groups $\mgroup_a,\ldots,\mgroup_b$ are replaced with a single
group $\mgroup^*$ spanning all their columns and~$P(\mgroup^*)$ is set to $\{\mrow\}$.
If the group $\mgroup_{a-1}$ ($\mgroup_{b+1}$ resp.) exists,
we insert $\mrow$ into $P(\mgroup_{a-1})$ ($P(\mgroup_{b+1})$)
only if this group is either in fact $\mgroup^{+}_{-}$
or is in category~\ref{c2}.
After such insertions, both invariants~\ref{i:psize} and~\ref{i:pmono}
may become violated.

Invariant~\ref{i:pmono} can only be violated if the group existed
$\mgroup_{-}^{+}$ and was not in category~\ref{c1}
and also there exists some other group with $\mrow$
in its potential row set.
Since it is impossible that $\mrow$ was inserted into potential
row sets of groups both to the left and to the right of $\mgroup_{-}^{+}$,
suppose wlog. that some $\mgroup'$ is to the right of $\mgroup_{-}^{+}$
and $r\in P(\mgroup')$.
In $O(\alpha)$ time we can check if $\mrow$ contains the column
minimum of the rightmost column of $\mgroup^+_-$ of $\mmat(\mcurrows\cup\{\mrow\},\mgroup^+_-)$.
If so, by Facts~\ref{f:monoton} and~\ref{f:contig}, we can delete
from $P(\mgroup^+_-)$ all the rows above $\mrow$
(recall that $\mrow$ contains a column minimum for the leftmost
column of $\mgroup'$).
Otherwise, by Fact~\ref{f:contig}, we can safely delete $\mrow$ 
from $P(\mgroup^+_-)$.
Hence, we fix invariant~\ref{i:pmono} in $O(\alpha(\log\log{m}+\log{\alpha}))$ time.

Invariant~\ref{i:psize} is violated if
$|P(\mgroup_{a-1})|=2\alpha$ or $|P(\mgroup_{b+1})|=2\alpha$.
In that case algorithm of Lemma~\ref{l:split} is used to split
group $\mgroup_{z}$, for $z\in\{a-1,b+1\}$ 
into  groups $\mgroup_z',\mgroup_z''$
such that $|P(\mgroup_z')|=|P(\mgroup_z'')|=\alpha$.

  \begin{figure}[t]
  \centering
  \includegraphics[scale=0.6]{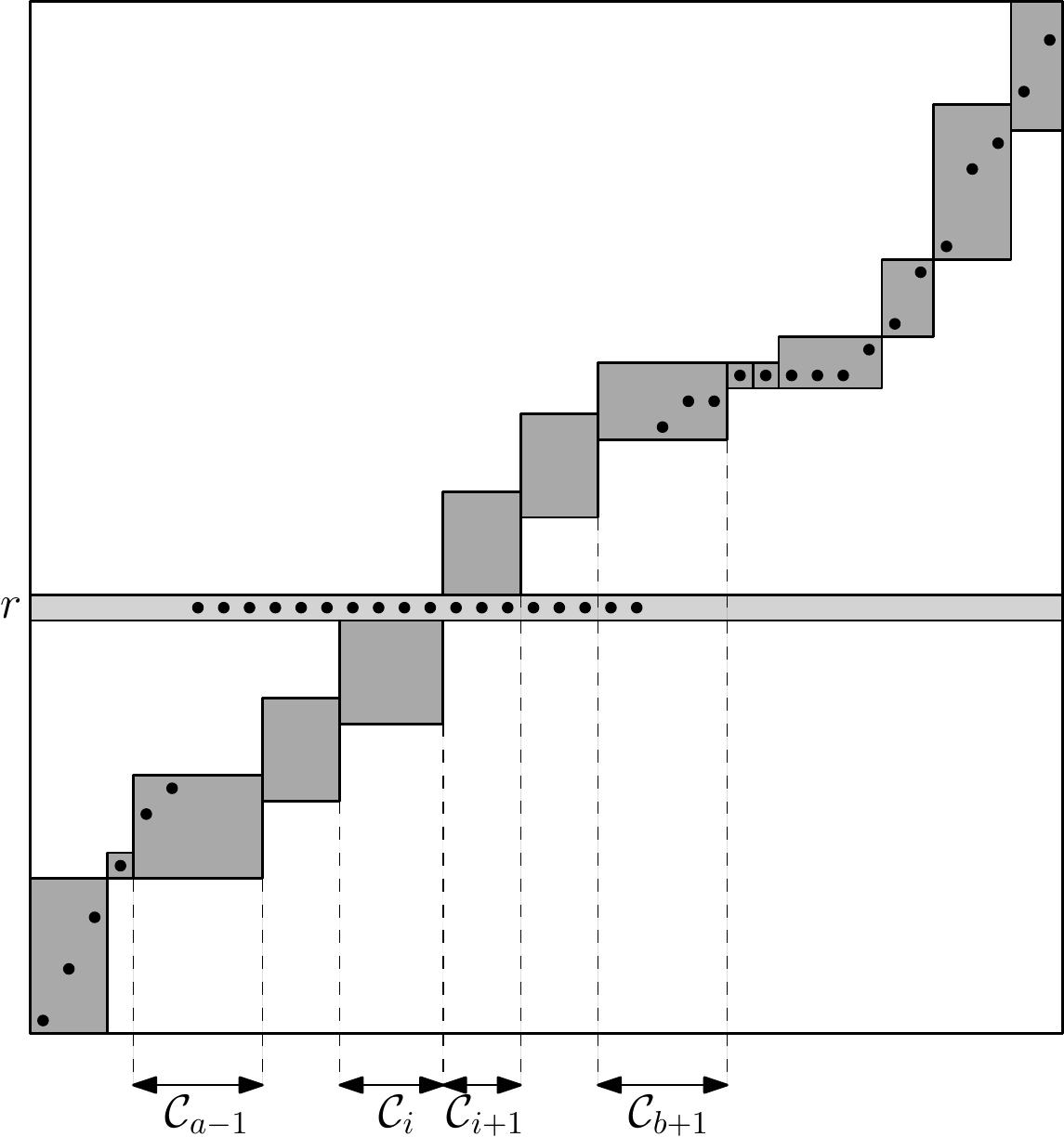}
  \caption{Updating the column groups and the corresponding
    potential row sets after activating row $\mrow$.
    The rectangles conceptually show the potential
    row sets.
    The rows of $\mcurrows$ that are not contained
    in any potential row set are omitted in the picture.
    The dots represent the column minima.
  Note that it might happen that $P(\mgroup_{i})$ contains rows both above and below $\mrow$.}
  \label{f:activate}
  \end{figure}

  We spend $O(\alpha(\log\log{\msize}+\log{\alpha}))$ time on identifying, accessing and updating
each group that falls into categories~\ref{c2} or~\ref{c3}.
There are $O(1)$ such groups, as discussed above.
Also, by Lemma~\ref{l:hinv}, it takes $O(\alpha\log\log{\msize})$ time to fix the
invariant~\ref{i:heap} for (possibly split) groups $\mgroup_{a-1}$,
$\mgroup_{b+1}$ and $\mgroup^*$.

In order to bound the running time of the remaining steps,
i.e., handling the groups of category~\ref{c1} and splitting
the groups that break the invariant~\ref{i:psize},
we introduce two types of credits for each element inserted into sets $P(\mgroup_i)$:
\begin{itemize}
  \item an $O(\log\log{\msize}+\log{\alpha})$ \emph{identification credit},
\item an $O\left(\frac{\log{m}}{\log{\alpha}}\right)$ \emph{splitting credit}.
\end{itemize}

The identification credit is used to pay for successfully verifying
that some group $\mgroup_i$ falls into category~\ref{c1} and deleting
all the elements of $P(\mgroup_i)$.
Indeed, as discussed above, we spend $O(|P(\mgroup_i)|(\log\log{\msize}+\log{\alpha}))$ time
on this.
As $P(\mgroup_i)$ is not empty, we can charge the cost of merging
$\mgroup_i$ with some other group to some arbitrary element of $P(\mgroup_i)$.
Recall that merging and splitting groups takes $O(\log\log{\msize})$ time.


Finally, consider performing a split of $P(\mgroup_i)$ of size $2\alpha$.
As the sets $P(\mgroup_i)$ only grow by inserting single elements,
there exist at least $\alpha$ elements of $P(\mgroup_i)$
that never took part in any split.
We use the total $O\left(\alpha\frac{\log{\msize}}{\log{\alpha}}\right)$ total credit
of those elements to pay for the split.

To sum up, the time needed to perform $\mrowcnt$ operations $\marow$
is \linebreak
$O\left(\mrowcnt\alpha(\log\log{\msize}+\log{\alpha})+I(\log\log{m}+\log{\alpha}+\frac{\log{\msize}}{\log{\alpha}})\right)$,
where $I$ is the total number of insertions to the sets $P(\mgroup_i)$.
As $\meb$ incurs $O(\mcolcnt)$ insertions in total, $I=O(\mrowcnt+\mcolcnt)$.
Setting $\alpha=\sqrt{\log{\msize}}$, we obtain the following lemma.

\begin{lemma}\label{l:full_report}
  Let $\mmat$ be a $\mrowcnt\times\mcolcnt$ offset Monge matrix.
  There exists a data structure supporting $\minit$ in $O(\mrowcnt+\mcolcnt\log{\msize})$ time,
$\mlb$ in $O(1)$ time and both $\mcolminrow$ and $\meb$ in
$O(\log{\msize})$ time.
Additionally, 
any sequence of $\marow$ operations is performed in
$O\left((\mrowcnt+\mcolcnt)\frac{\log{\msize}}{\log{\log{\msize}}}\right)$
total time, where $\msize=\max(\mrowcnt,\mcolcnt)$.
\end{lemma}

\section{Online Column Minima of a Block Monge Matrix}\label{s:block_min}
Let $\mmat=\om(\mmat_0,\off)$, $\mrows$, $\mcols$, $\mcolcnt$, $\mrowcnt, \msize$ be defined as in Section~\ref{s:full_min}.
In this section we consider the problem of reporting the column minima
of a rectangular offset Monge matrix, but in a slightly different setting.

Again, we are given a fixed rectangular Monge matrix $\mmat_0$ and we also
have an initially empty, growing set of rows $\mcurrows\subseteq\mrows$
for which the offsets $\off(*)$ are known.
Let $\Delta>0$ be an integral parameter not larger than~$\mcolcnt$.
We partition $\mcols$ into a set $\mblocks=\{\mblock_1,\ldots,\mblock_\blcnt\}$ of at most
$\lceil\mcolcnt/\Delta\rceil$ blocks, each of size at most $\Delta$.
The columns in each $\mblock_i$ constitute a contiguous fragment
of $\mcol_1,\ldots,\mcol_\mcolcnt$, and each block $\mblock_i$ is to the left of $\mblock_{i+1}$.
We also maintain a shrinking subset $\mbnotyet\subseteq\mblocks$ containing the
blocks $\mblock_i$, such that the minima $\min\{\mmat(\mrows,\mblock_i)\}$
are not yet known.
More formally, for each $\mblock_i\in \mblocks\setminus\mbnotyet$, we
have $\min\{\mmat(\mrows,\mblock_i)\}=\min\{\mmat(\mcurrows,\mblock_i)\}$.
Initially $\mbnotyet=\mblocks$.

For each column $\mcol$ not contained in any of the blocks of $\mbnotyet$,
the data structure explicitly maintains
the \emph{current minimum}, i.e., the value $\min\{\mmat(\mcurrows,\{\mcol\})\}$.
Moreover, when some new row is activated, the user is notified
for which columns of $\bigcup(\mblocks\setminus\mbnotyet)$ the
current minima have changed.

For blocks $\mbnotyet$,
the data structure only maintains the value $\min\{\mmat(\mcurrows,\bigcup\mbnotyet)\}$.
Once the user can guarantee that the value $\min\{\mmat(\mrows,\bigcup\mbnotyet)\}$
does not depend on the ``hidden'' rows $\mrows\setminus\mcurrows$,
the data structure can move a block $\mblock_i\in\mbnotyet$ such
that $\min\{\mmat(\mrows,\bigcup\mbnotyet)\}=\min\{\mmat(\mcurrows,\mblock_i)\}$
out of $\mbnotyet$ and make it possible to access
the current minima in the columns of $\mblock_i$.

More formally, we support the following set of operations:
\begin{itemize}
\item $\minit(\mrows,\mcols)$ -- initialize the data structure.
\item $\marow(\mrow)$, where $\mrow\in\mrows\setminus\mcurrows$ --
  add $\mrow$ to the set $\mcurrows$.
\item $\mblb()$ -- return $\min\{\mmat(\mcurrows,\bigcup\mbnotyet)\}$.
  If $\mcurrows=\emptyset$ or $\mbnotyet=\emptyset$, return $\infty$.
\item $\mbeb()$ -- tell the data structure that indeed
  $$\min\{\mmat(\mrows\setminus\mcurrows,\mcols)\}\geq\mblb()=\min\{\mmat(\mcurrows,\mblock_i)\},$$
  for some $\mblock_i\in\mbnotyet$, 
  i.e., the smallest element of $\mmat(\mrows,\bigcup\mbnotyet)$ does not
  depend on the entries of $\mmat$ located in rows $\mrows\setminus\mcurrows$.
  Again, it is the responsibility of the user to guarantee that this condition
  is in fact satisfied.

  As the minimum of $\mmat(\mrows,\mblock_i)$ can now be computed, $\mblock_i$
  is removed from $\mbnotyet$.
\item $\mcolmin(\mcol)$, where $\mcol\in\mcols$ --
  for $\mcol\in\bigcup(\mblocks\setminus\mbnotyet)$, return the explicitly maintained
  $\min\{\mmat(\mcurrows,\{\mcol\})\}$.
  For $\mcol\in\bigcup\mbnotyet$, set $\mcolmin(\mcol)=\infty$.
\end{itemize}

Additionally, the data structure provides an access to the queue $\mqueue$
containing
the columns $c\in\bigcup(\mblocks\setminus\mbnotyet)$ such that the most recent call
to either $\marow$ or \linebreak
$\mbeb$ resulted in a change (or an initialization,
if $\mcol\in\mblock_i$ and the last update was $\mbeb$, which moved $\mblock_i$
out of $\mbnotyet$)
of the value $\mcolmin(\mcol)$.

Note that there can be at most $\mrowcnt$ calls to $\marow$ and
no more than
 $\lceil\mcolcnt/\Delta\rceil$ calls to \linebreak
 $\mbeb$.

\subsection{The Components}
\paragraph{An Infrastructure for Short Subrow Minimum Queries}
In this section we assume that for any $\mrow\in\mrows$ and
$1\leq a,b\leq\mcolcnt$, $b-a+1\leq\Delta$,
it is possible to compute an answer to a subrow minimum
query $S(\mrow,a,b)$ (see Section~\ref{s:full_min}) on
matrix $\mmat_0$ (equivalently: $\mmat$) in constant time.
We call such a subrow minimum query \emph{short}.

\paragraph{The Block Minima Matrix} Define a $\mrowcnt\times\blcnt$ matrix $\mmat'$
with rows $\mrows$ and columns $\mblocks$, such that
$$\mmat'_{\mrow_i,\mblock_j}=\min\{\mmat(\{\mrow_i\},\mblock_j)\}.$$
As we assume that we can perform short subrow minima queries in $O(1)$ time,
and every block spans at most~$\Delta$ columns,
we can access the elements of $\mmat'$ in constant time.
Fact~\ref{f:condense} implies that $\mmat'$ is also
a rectangular Monge matrix.

We build the data structure of Section~\ref{s:full_min} for
matrix $\mmat'$.
For brevity we identify the matrix $\mmat'$ with this data
structure and write e.g. $\mmat'.\minit()$
to denote the call to $\minit$ of the data structure
built upon $\mmat'$.
This data structure handles the blocks contained in $\mbnotyet$.
\paragraph{The Exact Minima Array} For each column $\mcol\in\bigcup(\mblocks\setminus\mbnotyet)$,
the value $$\mcminarr(\mcol)=\min\{\mmat(\mcurrows,\{\mcol\})\}$$
is stored explicitly.
The operation $\mcolmin(\mcol)$ returns $\mcminarr(\mcol)$.

\paragraph{Rows Containing the Block Minima}
For each $\mblock_j\in(\mblocks\setminus\mbnotyet)$ we store the value
$$y_j=\mmat'.\mcolminrow(\mblock_j).$$
Note that the data structure of Section~\ref{s:full_min}
guarantees that for $\mblock_i,\mblock_j\in(\mblocks\setminus\mbnotyet)$
such that $i<j$, we have $y_i\geq y_j$.

The set of defined $y_j$'s grows over time.
We store this set in a dynamic predecessor/successor data structure $Y$.
We can thus perform insertions/deletions
and $\ppred$/$\psucc$ queries on a subset of $\{1,2,\ldots,\mrowcnt\}$
in $O(\log\log{\mrowcnt})=O(\log\log{\msize})$ time.

We also have two auxiliary arrays $\mfirstb$ and $\mlastb$
indexed with the rows of $\mrows$.
$\mfirstb(\mrow)$ ($\mlastb(\mrow)$)
contains the leftmost (rightmost respectively) block $\mblock_j$
such that $y_j=\mrow$.
Updating these arrays when $\mbnotyet$ shrinks is straightforward.

\paragraph{The Row Candidate Sets}
Two subsets $D_0$ and $D_1$ of $\mcurrows$ are maintained.
The set $D_q$ for $q=0,1$ contains the rows of $\mcurrows$
that may still prove useful when
computing the initial value of $\mcminarr(\mcol)$ for
$\mcol\in \bigcup\{\mblock_i : \mblock_i\in\mbnotyet \land i\bmod 2=q\}$.
For each such $c$, $D_q$ contains a row $\mrow$ such that
$\min\{\mmat(\mcurrows,\{\mcol\})\}=\mmat_{\mrow,\mcol}$.
Note that adding any row from $\mcurrows$ to $D_q$
does not break this invariant.
The call $\marow(r)$ always adds the row $r$ to both
$D_0$ and $D_1$.
The sets $D_q$ are stored in dynamic predecessor/successor
data structures as well.
\begin{remark}
There is a subtle reason why we keep
two row candidate sets $D_0$, $D_1$ responsible for
even and odd blocks respectively, instead of one.
Being able to separate two neighboring blocks
of each group with a block from the other group
will prove useful in an amortized analysis
of the operation $\mbeb$.
\end{remark}

\subsection{Implementing the Operations}

\paragraph{$\mbeb$}
The preconditions of this operation
ensure that it is valid to 
call \linebreak
$\mmat'.\meb()$, which in response returns some $\mblock_j$.
At this point we find the row~$y_j$ containing
the minimum of $\mmat(\mrows,\mblock_j)$ using
$\mmat'.\mcolminrow(\mblock_j)$.
The data structure~$Y$ and the arrays $\mfirstb$ and $\mlastb$
are updated accordingly.

As the block $\mblock_j$ is moved out of $\mbnotyet$, we need
to compute the initial values $\mcminarr(\mcol)$
for $\mcol\in\mblock_j$.
Let $y_j^-$ be the row returned
by $\mmat'.\mcolmin(\mblock_{j-1})$ if $j>0$ and $\mrow_1$ otherwise.
Similarly, set $y_j^+$ to be the row returned
by $\mmat'.\mcolmin(\mblock_{j+1})$ if $j<\blcnt$ and $\mrow_{\mrowcnt}$
otherwise.
Clearly, $y_j^-\geq y_j\geq y_j^+$.
First we prove that for each column $\mcol\in\mblock_j$, we have
$$\min\{\mmat(\mcurrows,\{\mcol\})\}=\min\{\mmat(\mcurrows\cap\{y_j^+,\ldots,y_j^-\},\{\mcol\})\},$$
that is, the search for the minimum in column $\mcol$ can be limited
to rows $y_j^+$ through~$y_j^-$.
By the definition of $\mmat'$, for some column $\mcol_j\in B_j$,
the minimum $\mmat(\mcurrows,\{\mcol_j\})$ is located in row $y_j$.
Now assume that $\mcol\in B_j$ is to the left of $\mcol_j$.
By Fact~\ref{f:monoton},
one minimum of $\mmat(\mcurrows,\{\mcol\})$ is located in
the rows of $\mcurrows$ (weakly) below $y_j$.
If $j>0$, then for some column $\mcol_j^-\in \mblock_{j-1}$
one minimum of $\mmat(\mcurrows,\{\mcol_j^-\})$ is located
in row $y_j^-$.
By Fact~\ref{f:monoton}, the minimum of $\mmat(\mcurrows,\{\mcol\})$
is located in rows (weakly) above~$y_j^-$.
Analogously we prove that for $\mcol\in \mblock_j$ to the right of $\mcol_j$,
the minimum is located in rows $y_j^+$ through~$y_j$.

We first add the rows $y_j^-,y_j^+$ to $D_{j\bmod 2}$.
From the definition of set $D_{j\bmod 2}$, for each column $\mcol\in \mblock_j$
it suffices to only consider the elements $\mmat_{\mrow,\mcol}$, where
$\mrow\in D_{j\bmod 2}\cap\{y_j^+,\ldots,y_j^-\}$ as potential minima
in column $\mcol$.
All such rows $\mrow$ can be found with $O(\log{\log{m}})$ overhead
per row using predecessor search on $D_{j\bmod 2}$.
Now we prove that after this step all such rows $\mrow$ except of
$y_j^-$ and $y_j^+$ can be safely removed from $D_{j\bmod 2}$.
Indeed, let $c_k$ be a column in some block $\mblock_k\in\mbnotyet$
such that $k\equiv j \pmod 2$ and $k<j$.
In fact, we have $k<j-1$.
By the Monge property, we have
$\mmat_{y_j^-,\mcol_k}+\mmat_{\mrow,\mcol_j^-}\leq\mmat_{\mrow,\mcol_k}+\mmat_{y_j^-,\mcol_j^-}$.
Also, from the definition of $y_j^-$, $\mmat_{y_j^-,\mcol_j^-}\leq \mmat_{\mrow,\mcol_j^-}$.
If we had $\mmat_{y_j^-,\mcol_k}>\mmat_{\mrow,\mcol_k}$, that would
lead to a contradiction.
Thus, $\mmat_{y_j^-,\mcol_k}\leq \mmat_{\mrow,\mcol_k}$,
and removing $\mrow$ from $D_{j\bmod 2}$ does not break
the invariant posed on $D_{j\bmod 2}$, as $y_j^-\in D_{j\bmod 2}$.
The proof of the case $k>j$ is analogous.

Let us now bound the total time spent on updating values $\mcminarr(\mcol)$
during the calls \linebreak
$\mbeb$.
For each column $\mcol\in\mblock_j$, all entries $\mmat_{\mrow,\mcol}$,
where $\mrow\in D_{j\bmod 2}\cap\{y_j^+,\ldots,y_j^-\}$ are
tried as potential column minima.
Alternatively, we can say that for each such row, we try to use
it as a candidate for minima of $O(\Delta)$ columns.
However, only two of these rows are not deleted from $D_{j\bmod 2}$
afterwards.
If we assign a credit of $\Delta$ to each row inserted into
$D_{j\bmod 2}$, this credits can be used to pay for considering
all the rows except of $y_j^-$ and $y_j^+$.
Thus, the total time spent on testing candidates for the minima
over all calls to $\mbeb$ can be bounded
by $O\left(\frac{\mcolcnt}{\Delta}\Delta+I\Delta\right)$, where $I$ is
the number of insertions to either~$D_0$ or $D_1$.
However, $I$ can be easily seen to be $O\left(\mrowcnt+\frac{\mcolcnt}{\Delta}\right)$
and thus the total number of candidates tried by $\mbeb$
is $O(\mrowcnt\Delta+\mcolcnt)$.
The total cost spent on maintaining and traversing sets $D_q$
is $O\left((I+\frac{\mcolcnt}{\Delta})\log{\log{\msize}}\right)=\
O\left((\mrowcnt+\frac{\mcolcnt}{\Delta})\log{\log{\msize}}\right)$.

\paragraph{$\marow$}
Suppose we activate the row $\mrow\in \mrows\setminus\mcurrows$.
The first step is to call \linebreak
$\mmat'.\marow(\mrow)$
and add $\mrow$ to sets $D_0$ and $D_1$.
The introduction of the row $\mrow$ may change
the minima of some columns $\mcol\in\bigcup(\mblocks\setminus\mbnotyet)$.
We now prove that there can be at most $O(\Delta)$ changes.
Recall that for each $\mblock_i\in \mblocks\setminus\mbnotyet$,
for some column $\mcol_i\in \mblock_i$ the minimum
of $\mmat(\mrows,\{\mcol_i\})$ is located in row $y_i$.
Note that $r\neq y_i$, as $r$ has just been activated.
Let $u$ be such that $y_u>\mrow$.
Then, for each block $\mblock_j\in\mblocks\setminus\mbnotyet$, where $j<u$,
Fact~\ref{f:monoton} implies that all the columns of $\mblock_j$ have their minima in rows below
$y_u$ (or exactly at $y_u$) and thus the introduction of row $\mrow$ does
not affect their minima.
Analogously, if $y_v<\mrow$, then the introduction of row $\mrow$
does not affect columns in blocks to the right of $B_v$.
Hence, $\mrow$ can only affect the exact minima in at most two blocks:
$B_u, B_v$, where $u=\mlastb(Y.\psucc(\mrow))$ and $v=\mfirstb(Y.\ppred(\mrow))$.
The blocks can be found in $O(\log{\log{m}})$ time,
whereas updating the values $\mcminarr(\mcol)$ (along with pushing
them to the queue $\mqueue$)
takes $O(\Delta)$ time.

Let us bound the total running time of any sequence of operations $\marow$
and \linebreak
$\mbeb$.
By Lemma~\ref{l:full_report}, the time spent on executing the data structure $\mmat'$
operations is
$O\left(\mrowcnt\frac{\log{\msize}}{\log{\log{\msize}}}\
+\frac{\mcolcnt}{\Delta}\log{\msize}\right)$
whereas the time spent on maintaining the predecessor
structures and updating the column minima
is $O\left(\mrowcnt\Delta+\mrowcnt\log\log{\msize}+\mcolcnt+\frac{\mcolcnt}{\Delta}\log\log{\msize}\right)$.
The following lemma follows.
\begin{lemma}\label{l:block}
  Let $\mmat=\om(\mmat_0,\off)$ be a $\mrowcnt\times\mcolcnt$ rectangular offset Monge matrix.
  Let $\Delta$ be the block size.
  Assume we can perform subrow minima queries spanning at most $\Delta$ columns of $\mmat_0$
  in $O(1)$ time.
  There exists a data structure supporting $\minit$ in $O(\mrowcnt+\mcolcnt+\frac{\mcolcnt}{\Delta}\log{\msize})$
time and both $\mblb$ and $\mcolmin$ in $O(1)$ time.
Any sequence of $\marow$ and $\mbeb$ operations is performed
in $O\left(\mrowcnt\left(\frac{\log{\msize}}{\log{\log{\msize}}}+\Delta\right)+\mcolcnt+\frac{\mcolcnt}{\Delta}\log{\msize}\right)$
time, where $\msize=\max(\mrowcnt,\mcolcnt)$.
\end{lemma}

\section{Online Column Minima of a Staircase Offset Monge Matrix}\label{s:stair_min}
In this section we show a data structure supporting
a similar set of operations as in Section~\ref{s:full_min}, but
in the case when the matrices $\mmat_0$ and $\mmat=\om(\mmat_0,\off)$ are staircase Monge matrices
with $\msize$ rows $\mrows=\{\mrow_1,\ldots,\mrow_m\}$
and $\msize$ columns $\mcols=\{\mcol_1,\ldots,\mcol_m\}$.
We still aim at reporting the column minima of $\mmat$,
while the set $\mcurrows$ of revealed rows
is extended and new bounds on $\min\{\mmat(\mrows\setminus\mcurrows,\mcols)\}$
are given.

In comparison to the data structure of Section~\ref{s:full_min},
we loosen the conditions posed on the operations $\mlb$ and
$\meb$.
Now, $\mlb$ might return a value smaller than $\min\{\mmat(\mcurrows,\mcurcols)\}$
and a single call to $\meb$ might not report any new column minimum
at all.
However, $\meb$ can still only be called if
$\min\{\mmat(\mrows\setminus\mcurrows,\mcols)\}\geq\mlb()$
and 
the data structure we develop in this section guarantees
that a bounded number of calls to 
$\meb$ suffices to report all the column minima of $\mmat$.

The exact set of operations we support is the following:
\begin{itemize}
\item $\minit(\mrows,\mcols)$ -- initialize the data structure and set
  $\mcurrows=\emptyset$ and $\mcurcols=\mcols$.
\item $\marow(\mrow)$, where $\mrow\in\mrows\setminus\mcurrows$ --
  add $\mrow$ to the set $\mcurrows$.
\item $\mlb()$ -- return a number $v$ such that $\min\{\mmat(\mcurrows,\mcurcols)\}\geq v$.
  If $\mcurrows=\emptyset$ or $\mcurcols=\emptyset$, return $\infty$.
\item $\meb()$ -- tell the data structure that we have
  $$\min\{\mmat(\mrows\setminus\mcurrows,\mcols)\}\geq \mlb().$$
  As for previous data structures, it is the responsibility of
  the user to guarantee that this condition is in fact satisfied.

  With this knowledge, the data structure may report some
  column $\mcol\in\mcurcols$ such that 
  $\min\{\mmat(\mrows,\{\mcol\})\}$ is known.
  However, it's also valid to not report any new column minimum
  (in such case $\nil$ is returned) and only change
  the known value of $\mlb()$.
\item $\mcolmin(\mcol)$, where $\mcol\in\mcols$ --
  if $\mcol\in\mcols\setminus\mcurcols$, return the known minimum
  in column $\mcol$.
  Otherwise, return $\infty$.
\end{itemize}

\subsection{Partitioning a Staircase Matrix into Rectangular Matrices}
Before we describe the data structure, we prove the following
lemma on partitioning staircase matrices into rectangular
matrices.
\begin{lemma}\label{l:part}
For any $\epsilon\in(0,1)$,
a staircase matrix $\mmat$ with $\msize$ rows and $\msize$ columns can be 
partitioned in $O(\msize)$ time into $O(\msize)$ non-overlapping
rectangular matrices so that each
row appears in $O\left(\frac{\log{\msize}}{\log{\log{\msize}}}\right)$
matrices of the partition, whereas each column appears
in $O\left(\frac{\log^{1+\epsilon}{\msize}}{\log{\log{\msize}}}\right)$
matrices of the partition.
\end{lemma}

\newcommand{\pst}{\mmat'^{\text{s}}}
\newcommand{\pre}{\mmat'^{\text{r}}}

\begin{proof}
  Let $\epsilon\in(0,1)$ and set $b=\lfloor\log^{\epsilon}{m}\rfloor$.
  For $m>1$, we have $b\geq 1$.

  We first describe the partition for matrices $\mmat'$ with $\msize'=b^z$ rows
  $\{\mrow_1',\ldots,\mrow_{\msize'}'\}$
  and $\msize'$ columns
  $\{\mcol_1',\ldots,\mcol_{\msize'}'\}$,
  where $z\geq 0$.
  Our partition will be recursive.
  If $z=0$, then $\mmat'$ is a $1\times 1$ matrix and our partition
  consists of a single element $\mmat'$.

  Assume $z>0$.
  We partition $\mmat'$ into $b$ staircase matrices $\pst_1,\ldots,\pst_b$
  and $b-1$ rectangular matrices $\pre_1,\ldots,\pre_{b-1}$.
  For $i=1,\ldots,b$, we set the $i$-th staircase matrix to be
  $$\pst_i=\mmat'(\{\mrow_{(i-1)b^{z-1}+1}',\ldots,\mrow_{ib^{z-1}}'\},\
  \{\mcol_{(i-1)b^{z-1}+1}',\ldots,\mcol_{ib^{z-1}}'\}),$$
  whereas for $j=1,\ldots,b-1$, the $j$-th rectangular matrix is defined as
  $$\pre_i=\mmat'(\{\mrow_{(i-1)b^{z-1}+1}',\ldots,\mrow_{ib^{z-1}}'\},\
  \{\mcol_{ib^{z-1}+1}',\ldots,\mcol_{\msize'}'\}).$$
  See Figure~\ref{f:rec_part} for a schematic depiction of such partition.

  \begin{figure}[t]
  \centering
  \includegraphics[scale=0.6]{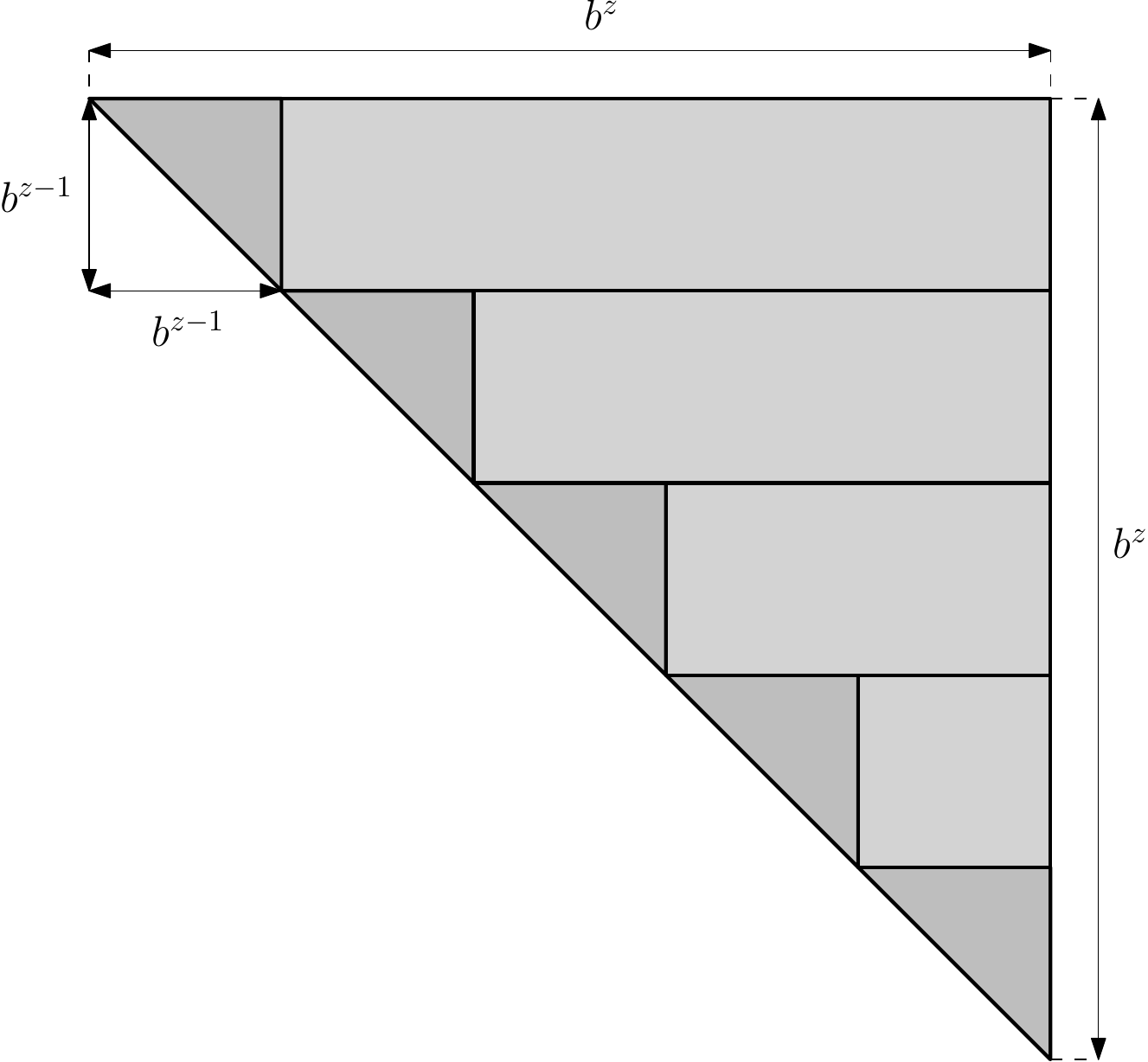}
  \caption{Schematic depiction of the biased partition used in Lemma~\ref{l:part}.}
  \label{f:rec_part}
  \end{figure}

  Each of matrices $\pst_i$ is of size $b^{z-1}\times b^{z-1}$ and 
  is then partitioned recursively.

  Let us now compute the value $\prcount(z)$ ($\pccount(z)$) defined 
  as the maximum number of matrices in partition that some given row (column resp.)
  of a staircase matrix $\mmat'$ of size $b^z\times b^z$ appears in.
  Clearly, on the topmost level of recursion,
  each row appears in exactly one staircase matrix~$\pst_i$
  and at most one rectangular matrix $\pre_j$.
  Thus, we have $\prcount(0)=1$ and $\prcount(z+1)\leq\prcount(z)+1$,
  which easily implies $\prcount(z)\leq z+1$.
  
  Each column appears in exactly one matrix $\pst_i$ and no more than
  $b-1$ matrices $\pre_j$.
  Hence, we have $\pccount(0)=1$ and $\pccount(z+1)\leq \pccount(z)+b-1$.
  We thus conclude $\pccount(z)\leq zb-z+1$.

  Analogously we can compute the value $\pfcount(z)$ denoting the total number
  of rectangular matrices in such recursive partition.
  We have $\pfcount(0)=1$ and $\pfcount(z+1)\leq b\cdot\pfcount(z)+b-1$.
  An easy induction argument shows that $\pfcount(z)\leq 2b^z-1$.

  The partition for an arbitrary matrix $\mmat$ of $\msize$ is obtained as follows.
  We find the smallest $y$ such that $b^y\geq \msize$.
  We next find the recursive partition of matrix $\mmat^*$ which is
  defined as $\mmat$ padded so that it has $b^y$ rows and $b^y$ columns.
  The last step is to remove some number of dummy rightmost columns
  and bottommost rows from each rectangular matrix of the partition.
  
  Now, each row of $\mmat$ appears in at most
  $$y+1=O(\log_b{\msize})=O\left(\frac{\log{\msize}}{\log{b}}\right)=\
  O\left(\frac{\log{\msize}}{\epsilon\log{\log{m}}}\right)=\
  O\left(\frac{\log{\msize}}{\log{\log{m}}}\right)$$
  rectangular matrices of a partition.
  Each column of $\mmat$ appears in at most
  $$yb-y+1\leq yb+1=O(b\log_b{\msize})=O\left(\frac{b\log{\msize}}{\log{b}}\right)=\
  O\left(\frac{\log^{1+\epsilon}{\msize}}{\epsilon\log{\log{m}}}\right)=\
  O\left(\frac{\log^{1+\epsilon}{\msize}}{\log{\log{m}}}\right)$$
  matrices of the partition.
  The partition consists of at most $2b^y-1=O(m)$ rectangles.
  The time needed to compute the row and column intervals
  constituting the rows and columns of the individual matrices
  of the partition is $O(\pfcount(y))=O(m)$.
\end{proof}

\subsection{The Data Structure Components}
\paragraph{The Short Subrow Minimum Queries Infrastructure}
Let $\Delta=\lceil\log^{1-\epsilon/2}{\msize}\rceil$.
In order to be able to use the data structure of Lemma~\ref{l:block} with
block size $\Delta$, we need the following lemma.
\begin{lemma}\label{l:short}
The staircase Monge matrix $\mmat_0$ can be preprocessed in $O(\msize\Delta\log{\msize})$ time
to allow performing subrow minimum queries on $\mmat_0$ spanning at most
$\Delta$ columns in $O(1)$ time.
\end{lemma}
\begin{proof}
We use the following result of $\cite{Gawry:2015}$.
\begin{lemma}[Lemma 3. of \cite{Gawry:2015}]\label{l:entrow}
Given a $y\times x$ rectangular Monge matrix $\mmat'$, one can construct
in $O(x\log{y})$ time an $O(x)$-space data structure supporting
subrow minimum queries spanning \emph{all} columns of $\mmat'$ in $O(1)$ time, if $x=O(\log{y})$.
\end{lemma}
Let $q$ be the maximum integer such that $2^q<\Delta$.
For $j\in[0,q]$ and $i\in [1,\msize-2^j+1]$, let $\mmat^j_i=\mmat(\{\mrow_1,\ldots,\mrow_i\},\{\mcol_i,\ldots,\mcol_{i+2^j-1}\})$,
i.e., $\mmat^j_i$ is a rectangular submatrix of $\mmat$ with columns
$\{\mcol_i,\ldots,\mcol_{i+2^j-1}\}$ and all rows that have
values defined for these columns.
By Fact~\ref{f:submatrix}, $\mmat^j_i$ is a Monge matrix.
For each $\mmat^j_i$, we build a data structure of Lemma~\ref{l:entrow}.
This takes
$$O\left(\sum_{j=0}^q\sum_{i=1}^{\msize-2^j+1}2^j\log{\msize}\right)=O\left(\sum_{j=0}^q2^j\msize\log{\msize}\right)=\
O(2^q\msize\log{\msize})=O(\msize\Delta\log{\msize})$$
time.
Now we show how to handle a subrow minimum query $S(\mrow,a,b)$ on $\mmat_0$, where
$b-a+1\leq\Delta$.
Let $u$ be the greatest integer such that $2^u\leq b-a+1$.
Then we can cover our subrow minimum query with two possibly overlapping
queries of length $2^u$.
Hence, to answer $S(\mrow,a,b)$ it is enough to find
the minimum in row $\mrow$ in $\mmat^j_{a}$ and the minimum
in row $\mrow$ in $\mmat^j_{b-2^j+1}$ and return the smaller one.
By Lemma~\ref{l:entrow}, this takes $O(1)$ time.
\end{proof}

\paragraph{The Partition $\mmat_1,\ldots,\mmat_q$}
  We partition the staircase Monge matrix $\mmat$ into $O(\msize)$ \emph{rectangular} Monge matrices
  $\mmat_1,\ldots,\mmat_q$ such that each $\mmat_i$ is a subrectangle
  of $\mmat$.
  By Lemma~\ref{l:part}, we can ensure that each row appears in
  $O\left(\frac{\log{\msize}}{\log{\log{\msize}}}\right)$ subrectangles and
  each column appears in $O\left(\frac{\log^{1+{\epsilon/2}}{\msize}}{\log{\log{\msize}}}\right)$
  subrectangles.
  Every element of $\mmat$ is covered by exactly one matrix $\mmat_i$.
  For each row $\mrow$ (column~$\mcol$) we compute the set $W_\mrow$ ($W^\mcol$ respectively)
  of matrices of the partition in which $\mrow$ ($\mcol$ resp.) appears.

  We build the block data structure of Section~\ref{s:block_min} for each $\mmat_i$.
  For each $\mmat_i$ we use the same block size $\Delta$.
  As each $\mmat_i$ is a subrectangle of $\mmat$, Lemma~\ref{l:short} guarantees that
  we can perform subrow minimum queries on $\mmat_i$ spanning at most
  $\Delta$ columns in $O(1)$ time.
  For brevity, we identify the matrix $\mmat_i$
  and its associated data structure.
  We use the dot notation to denote operations acting on specific
  matrices, e.g. $\mmat_i.\minit$.

  For each matrix $\mmat_i$ we use notation analogous as in previous sections:
  $\mrows_i$ and $\mcols_i$ are the sets of rows and columns of~$\mmat_i$,
  respectively.
  Let $\mrowcnt_i=|\mrows_i|$ and $\mcolcnt_i=|\mcols_i|$.
  Denote by $\mcurrows_i$ the set of active rows of $\mmat_i$.
  Recall that the blocks of the matrix $\mmat_i$ are
  partitioned into two sets $\mbnotyet_i$ and $\mblocks_i\setminus\mbnotyet_i$.
  Denote by $\bpart(\mmat_i)$ the submatrix $\mmat_i(\mcurrows_i,\bigcup\mbnotyet_i)$
  and by $\epart(\mmat_i)$ the submatrix $\mmat_i(\mcurrows_i,\bigcup(\mblocks_i\setminus\mbnotyet_i))$.

  \paragraph{The Priority Queue $H$}
  The core of our data structure is a priority queue $H$.
  At any time, $H$ contains an element $\mcol$ for each column
  $\mcol\in\mcurcols$ and at most one element $\mmat_i$ for
  each matrix $\mmat_i$.
  Thus the size of $H$ never exceeds $O(m)$.

  We maintain the following invariants after $\minit$ and each call
  $\marow$ or \linebreak
  $\meb$ resulting in $\mcurcols\neq\emptyset$:
  \begin{enumerate}[label=H.\arabic*,leftmargin=*]
    \item For each $\mcol\in\mcurcols$, the key of $\mcol$ in $H$ is\label{i:col}
      equal to $$\min\{\mmat_i.\mcolmin(\mcol):\mmat_i\in W^\mcol\}.$$
    \item For each $\mmat_i$ such that $\bpart(\mmat_i)$ is not empty,
      the key of $\mmat_i$ in $H$ is equal to\label{i:block}
      $$\min\{\bpart(\mmat_i)\}=\mmat_i.\mblb().$$
  \end{enumerate}
  \begin{lemma}\label{l:hlb}
    Assume invariants~\ref{i:col} and~\ref{i:block} are satisfied.
    Then $H.\pqminkey()\leq \mmat(\mcurrows,\mcurcols)$.
  \end{lemma}
  \begin{proof}
  Let $v=H.\pqminkey()$. Assume the contrary, that there exists an element $\mmat_{\mrow,\mcol}<v$, where
  $\mrow\in\mcurrows$ and $\mcol\in\mcurcols$.
  Let $\mmat_i$ be the rectangular Monge matrix containing element $\mmat_{\mrow,\mcol}$.
  If $\mcol\in\bigcup\mbnotyet_i$, then $v>\mmat_{\mrow,\mcol}\geq \mmat_i.\mblb()$.
  But then the key of $\mmat_i$ in $H$ is $\mmat_i.\mblb()$, a contradiction.
  Similarly, if $\mcol\in\bigcup(\mblocks_i\setminus\mbnotyet)$, then
  $\mmat_i.\mcolmin(\mcol)\leq \mmat_{\mrow,\mcol}<v$, a contradiction.
  \end{proof}
  \subsection{Implementing the Operations}
  \paragraph{Initialization} The procedure $\minit$ first builds
  the short subrow minimum query data structure of Lemma~\ref{l:short}.
  Then, the data structure of Lemma~\ref{l:block} is initialized for each $\mmat_i$.
  The total time needed to initialize these structures
  is thus
  $$O\left(\msize\Delta\log{\msize}+\
    \msize\frac{\log^{1+\epsilon/2}{\msize}}{\log{\log{\msize}}}+
    \frac{\msize}{\Delta}\frac{\log^{2+\epsilon/2}{\msize}}{\log{\log{\msize}}}\
    \right)=\
    O\left(\msize\log^{2-\epsilon/2}{\msize}\right).$$
  Next, $\minit$ inserts into the priority queue
  $H$ an element $\mcol$ with key $\infty$ for each $\mcol\in\mcols$
  and an element $\mmat_i$ with key $\infty$ for each matrix $\mmat_i$.
  This takes additional $O(\msize)$ time.
  Clearly, the invariants~\ref{i:col} and~\ref{i:block} are satisfied immediately after the initialization.
    
  \paragraph{\mlb} By Lemma~\ref{l:hlb}, the value $v=H.\pqminkey()$ is a lower bound
  on the value $\min\{\mmat(\mcurrows,\mcurcols)\}$.
  The function $\mlb$ returns $v$ and thus works in $O(1)$ time.

  \paragraph{\marow} The call $\marow(\mrow)$ may require changes to some keys
  of the entries of $H$ in order to satisfy invariants~\ref{i:col} and~\ref{i:block}.
  However, the activation of $\mrow$ does not alter what the functions $\mmat_i.\mcolmin(\mcol)$ or
  $\mmat_i.\mblb()$ return for matrices $\mmat_i\notin W_\mrow$.
  For all $\mmat_i\in W_\mrow$ we call $\mmat_i.\marow(\mrow)$.
  By Lemma~\ref{l:block}, the columns $\mcol_j$ of $\epart(\mmat_i)$ with changed minima
  can be read in linear time from $\mmat_i.\mqueue$.
  If ${\mcol_j\in \mcurcols}$ and the current key of~$\mcol_j$
  in $H$ is greater than $\mmat_i.\mcolmin(\mcol_j)$,
  we decrease key of~$c_j$ in~$H$.
  Analogously, the call $\marow(\mrow)$ can incur the change of $\mmat_i.\mblb()$
  and thus we may need to decrease the key of $\mmat_i$ in $H$.
  In both cases, as the operation
  $H.\pqdec$ runs in $O(1)$ time,
  the time spent on decreasing keys in $H$ is asymptotically no more
  than the running time of $\mmat_i.\marow(\mrow)$,
  
  \paragraph{\meb}
  Let $v=\mlb()=H.\pqminkey()$.
  Recall that the precondition of $\meb$ requires
  $\min\{\mmat(\mrows\setminus\mcurrows,\mcurcols)\}\geq\min\{\mmat(\mrows\setminus\mcurrows,\mcols)\}\geq v$.
  Also, $\min\{\mmat(\mcurrows,\mcurcols)\}\geq v$,
  so we can conclude $\min\{\mmat(\mrows,\mcurcols)\}\geq v$.
  We have two cases.

  \noindent{\textbf{1.}} If the top element of $H$ is a column $\mcol$, then from invariant~\ref{i:col} we know that
  $\mcol\in\mcurcols$ and for some $\mmat_j\in W^\mcol$ we have:
  $$\min\{\mmat(\mrows,\mcurcols)\}\geq\
  v=\mmat_j.\mcolmin(\mcol)\geq\min\{\mmat(\mrows,\{\mcol\})\}.$$
  However, clearly $\min\{\mmat(\mrows,\{\mcol\})\}\geq\min\{\mmat(\mrows,\mcurcols)\}$,
  so we conclude that the inequalities
  are in fact equalities and $v$ is indeed the minimum in column $\mcol$.
  In that case $\mcol$ is returned by $\meb$ and $\mcol$ is
  removed from $\mcurcols$.
  It can be easily verified that after calling $H.\pqext()$ invariants~\ref{i:col}
  and~\ref{i:block} still hold.
  This case arises at most once for each column of $\mcols$ so the
  total cost of $H.\pqext$ calls for all columns is $O(\msize\log{\msize})$.

  \noindent{\textbf{2.}} Now consider the case when the top element of $H$
  is a matrix $\mmat_i$.
  In this case we return $\nil$ and do not alter the set $\mcurcols$.
  As $\mmat_i$ is a subrectangle of $\mmat$, from the precondition we have
  $$\min\{\mmat_i(\mrows_i\setminus\mcurrows,\mcols_i)\}\geq\min\{\mmat(\mrows\setminus\mcurrows,\mcols)\}\geq v=\
  \mmat_i.\mblb().$$
  Hence, we can call $\mmat_i.\mbeb()$.
  Recall that this operation shrinks the set~$\mbnotyet_i$
  and thus we need to update $H$ so that the invariants~\ref{i:col} and~\ref{i:block} are
  satisfied.
  First we pop the entry $\mmat_i$ from $H$ with $H.\pqext()$ in $O(\log{\msize})$ time.
  Now, if $\mbnotyet_i\neq\emptyset$, we once again need to insert into $H$
  an element~$\mmat_i$ with key $\mmat_i.\mblb()$ in order to satisfy
  invariant~\ref{i:block}.
  To satisfy invariant~\ref{i:col}, we decrease key of each $c_j\in \mmat_i.\mqueue$
  to $\mmat_i.\mcolmin(c_j)$, if appropriate.
  Again, as decreasing a key in $H$ takes constant time, the time spent
  on decreasing column keys is asymptotically the same
  as the cost of the recent call to $\mmat_i.\mbeb$.

  A call to $\mmat_i.\mbeb$ can happen at most $O(\mcolcnt_i/\Delta)$
  times, so the additional time spent on~$H$ operations
  incurred by the calls to $\mmat_i.\mbeb$ is $O((\mcolcnt_i/\Delta)\cdot\log{m})$.
  For the same reason, the call $\meb$ returns $\nil$ at most
  $O\left(\sum_i\mcolcnt_i/\Delta\right)$ times.
  The total number of calls to $\meb$ to compute all the column
  minima of $\mmat$ is thus $O\left(\msize+\sum_i\mcolcnt_i/\Delta)=O(\msize\log^{\epsilon}{\msize}\right)$.
  The total cost of operations on $H$ that were not
  charged to $\mmat_i.\mbeb$ calls is $O(\msize\log^{1+\epsilon}{\msize})$.

  Let us now compute the total time spent in the calls
  $\mmat_i.\marow$ and \linebreak
  $\mmat_i.\mbeb$.
  
  $$\sum_i O\left(\mrowcnt_i\left(\Delta+\frac{\log{\msize}}{\log{\log{\msize}}}\right)+\mcolcnt_i+\frac{\mcolcnt_i}{\Delta}\log{\msize}\right)=\
  O\left(\frac{\log{\msize}}{\log{\log{\msize}}}\sum_i\mrowcnt_i+\log^{\epsilon/2}{\msize}\sum_i\mcolcnt_i\right)=$$
  $$O\left(\msize\left(\frac{\log{\msize}}{\log{\log{\msize}}}\right)^2+m\log^{1+\epsilon}{\msize}\right)=\
  O\left(\msize\left(\frac{\log{\msize}}{\log{\log{\msize}}}\right)^2\right).$$
  
\begin{lemma}\label{l:staircase_report}
  Let $\mmat=\om(\mmat_0,\off)$ be a $\msize\times\msize$ offset staircase Monge matrix and let
  $\epsilon\in(0,1)$.
There exists a data structure supporting $\minit$ in
$O\left(\msize\log^{2-\epsilon}{\msize}\right)$
time and both $\mlb$ and $\mcolmin$ in $O(1)$ time.
Any sequence of $\marow$ and\linebreak
$\meb$ operations can be performed
in $O\left(\msize\left(\frac{\log{\msize}}{\log{\log{\msize}}}\right)^2\right)$ time.
All the column
minima are computed after $O(\msize\log^{\epsilon}{\msize})$ calls
to 
$\meb$.

\end{lemma}

\begin{remark}\label{r:flipped_report}
Lemma~\ref{l:staircase_report} also holds for flipped staircase matrices.
\end{remark}
\begin{proof}
A flipped staircase matrix $\mmat'$ can be seen as a staircase matrix $\mmat$
with both the rows and columns reversed.
Each subrow minimum query on $\mmat'$ translates easily into a single
subrow minimum query on~$\mmat$.
\end{proof}

\section{Single-Source Shortest Paths in Dense Distance Graphs}\label{s:dijkstra}
In this section we study the possibly most general instance
of the problem of computing single-source shortest paths in dense
distance graphs, that fits all the most important applications.
The overall structure of our algorithm resembles Dijkstra's algorithm
and is also similar to both \cite{Fak:2006} and \cite{Mozes:2014}.
Nevertheless, we give a complete implementation and analysis.

Let $G=(V,E)$ be a weighted planar digraph and let
$G_1,\ldots,G_g$ be some partition of $G$ into \emph{connected} regions with few holes.
Denote by $X_{i,j}$ the vertices of $\bnd{G_i}$ lying
on the $j$-th hole of $G_i$, in clockwise order.
Also assume that for each $u,v\in\bnd{G_i}$ there exists
a path in $G_i$ -- each region could be easily extended
with bidirectional copies of edges of $G_i$ with
some very large length so that we can tell if the path
actually exists by only looking at the weight of the shortest path.

The graph $G$, the partition $G_1,\ldots,G_g$
and the dense distance graphs $\ddg(G_i)$ constitute
the ,,fixed input'' of our problem.
We are allowed
to preprocess each $\ddg(G_i)$ once in time asymptotically
no more than the time used for construction of $\ddg(G_i)$.
To the best of our knowledge, in all known applications
this time is no less than $O((|V(G_i)|+|\bnd{G_i}|^2)\log{|V(G_i)|})$, which is the running time of the method
described in Section~\ref{s:prelim}.
Denote by $\ddg(G_i)[x,y]$ the length of the shortest path
$x\to y$ in~$G_i$.

Now, let $P$ be some set of ``outer'' directed edges with both
endpoints in $\bnd{G}$,
not necessarily contained in~$E$ and
not necessarily preserving the planarity of $G$.
Denote by $\ell(e)\in\mathbb{R}$ the length of edge $e\in P$.
Also, let $\phi$ be a function $\bnd{G}\to \mathbb{R}$, called a \emph{price function}.
We define reduced lengths with respect to $\phi$ for both edges of $P$ and
distances in dense distance graphs.
\begin{itemize}
  \item for $(x,y)\in P$ let $\ell^\phi((x,y))=\ell((x,y))+\phi(x)-\phi(y)$,
  \item for $(x,y)\in \bnd{G_i}\times \bnd{G_i}$, let $\ddg(G_i)^\phi[x,y]=\ddg(G_i)[x,y]+\phi(x)-\phi(y)$.
\end{itemize}
A price function is called \emph{feasible}, if all the reduced lengths
are non-negative.
In all the relevant applications we also assume that for each $G_i$ we are
given (as part of the ,,fixed'' input) a feasible price
function~$\phi_i^0$.
\begin{remark}
In the flow-related applications (\cite{Borradaile:2011, Lacki:2012}),
graphs $\ddg(G_i)$ typically contain non-negative lengths, and hence $\phi_i^0\equiv 0$.
The distance oracles (\cite{Kaplan:2012, Mozes:2012})
typically handle negative edges during their initialization.
Actually, any algorithm following the original recursive
approach of \fak and Rao \cite{Fak:2006}
to construct dense distance graphs (and simultaneously compute single-source
shortest paths in the case of real edge lengths)
can be extended to find the feasible $\phi_i^0$ for each $G_i$.
\end{remark}

Given the set $P$,
the function $\ell$, a feasible $\phi$ and a vertex $s\in \bnd{G}$, our goal is to
design an efficient subroutine
computing the lengths
of the shortest paths from $s$ to all vertices of $\bnd{G}$
in graph $(\bnd{G},P)\cup \ddg(G)$,
assuming edge-lengths reduced by $\phi$.
As the reduction of lengths does not change the shortest paths,
we can follow the general approach of Dijkstra's algorithm,
which assumes non-negative edge lengths.
However, our subroutine has to be robust enough to
not preprocess the \emph{entire} graph $\ddg(G)$ each time the subroutine is called
with different parameters $\phi$ and $P$.
We call this problem the \emph{single-source shortest paths in a dense distance graph}
problem.

Such a subroutine has been used e.g. in an $O(n\log^3{n})$ time algorithm of 
Borradaile et al. \cite{Borradaile:2011}, computing
the multiple-source multiple-sink maximum flow in a directed planar
graphs.
Their subroutine extends FR-Dijkstra \cite{Fak:2006} to work with reduced lengths
for the case of a single-region graph $G$ with distinguished boundary
vertices $\bnd{G}$ lying on a single face of $G$ and a set of ``outer'' edges $P$.
The computation of the lengths of shortest paths from
$s\in \bnd{G}$ to all vertices of $\bnd{G}$
takes $O(|\bnd{G}|\log^2{|\bnd{G}|}+|P|\log{|\bnd{G}|})$ time and is a bottleneck of their algorithm.
In this section we propose a more efficient implementation.

\subsection{The Algorithm}
\begin{lemma}\label{l:decomp}
  Suppose a feasible price function $\phi_i^0$ for $\ddg(G_i)$ is given.
  The graph $\ddg(G_i)$ can be decomposed into a set of $O(1)$ (flipped) staircase
  Monge matrices $D_i$ of at most $|\bnd{G_i}|$ rows and columns.
  For each $u,v\in \bnd{G_i}$ we have:
  \begin{itemize}
  \item for each $\mmat\in D_i$ such that
  $\mmat_{u,v}$ is defined, $\mmat_{u,v}\geq \ddg(G_i)[u,v]$.
  \item there exists $\mmat\in D_i$ such that
  $\mmat_{u,v}$ is defined and $\mmat_{u,v}=\ddg(G_i)[u,v]$.
  \end{itemize}
  The decomposition can be computed in $O((|V(G_i)|+|\bnd{G_i}|^2)\log{|V(G_i)|})$ time.
\end{lemma}
\begin{proof}

  Let $h$ be the number of holes of $G_i$.
  We describe the set of (flipped) staircase Monge matrices $D_i$.

  First, for the $j$-th hole we add to $D_i$ a staircase
  matrix $\mmat^{j+}$ and a flipped staircase matrix~$\mmat^{j-}$ with
  rows $X_{i,j}$ and columns $X_{i,j}$.
  The order imposed on the rows and columns is the clockwise order
  on the $j$-th hole.
  For $u,v\in X_{i,j}$, $u\leq v$, we set $\mmat^{j+}_{u,v}=\ddg(G_i)^\phi[u,v]$.
  For $u\geq v$, we set $\mmat^{j-}_{u,v}=\ddg(G_i)^\phi[u,v]$.
  The matrices $\mmat^{j+}$ and $\mmat^{j-}$ represent the distances
  between the nodes of a single hole.
  We now prove that both $\mmat^{j+}$ and $\mmat^{j-}$ are Monge.
  Let $v,x,y,z$ be some nodes of $X_{i,j}$ in clockwise order.

  Assume $\mmat^{j+}_{v,y}+\mmat^{j+}_{x,z}<\mmat^{j+}_{v,z}+\mmat^{j+}_{x,y}$, or,
  equivalently, $\ddg(G_i)[v,y]+\ddg(G_i)[x,z]<\ddg(G_i)[v,z]+\ddg(G_i)[x,y]$.
Recall that the matrix $\ddg(G_i)$ represents distances between
all pairs of vertices of $\bnd{G_i}$ in $G_i$.
As the vertices of $X_{i,j}$ lie on a single face of a planar graph $G_i$,
any path $v\to y$ in $G_i$ has to cross each path $x\to z$
in $G_i$.
Specifically, the shortest path $v\xrightarrow{p_1}y$ and the shortest path
$x\xrightarrow{p_2}z$ have some common vertex $u\in G_i$.
Thus, the total length of paths $v\xrightarrow{p_1}u\xrightarrow{p_1} y$ and
$x\xrightarrow{p_2} u\xrightarrow{p_2} z$
is $\ddg(G_i)[v,y]+\ddg(G_i)[x,z]$.
But the paths $v\xrightarrow{p_1}u\xrightarrow{p_2} z$ and
$x\xrightarrow{p_2} u\xrightarrow{p_1} y$ also have the
same total length and that length cannot be less than $\ddg(G_i)[v,z]+\ddg(G_i)[x,y]$.
This contradicts
$\ddg(G_i)[v,y]+\ddg(G_i)[x,z]<\ddg(G_i)[v,z]+\ddg(G_i)[x,y]$
and thus
proves that $\mmat^{j+}$ is a staircase Monge matrix.
The proof that $\mmat^{j-}$ is a flipped staircase Monge matrix is analogous.
Both $\mmat^{j+}$ and $\mmat^{j-}$ are computed in $O(|\bnd{G_i}|^2)$ time.

Now we describe the matrices of $D_i$ representing the distances between
nodes lying on the $j$-th hole and nodes on the $k$-th hole (for $j\neq k$).
Mozes and Wulff-Nilsen (\cite{Mozes:2010}, Section~4.4) showed that
given the initial feasible price function $\phi_i^0$, in
$O((|V(G_i)|+|\bnd{G_i}|^2)\log{|V(G_i)|})$ time one can compute
two rectangular Monge matrices $\mmat^{j,k,\text{L}}$ and
$\mmat^{j,k,\text{R}}$, with rows $X_{i,j}$ and columns $X_{i,k}$,
such that for each $u\in X_{i,j}$
and $v\in X_{i,k}$ we have
$\ddg(G_i)[u,v]=\min(\mmat^{j,k,\text{L}}_{u,v},\mmat^{j,k,\text{R}}_{u,v})$
(for more details about this construction, see also \cite{Kaplan:2012}, Section~5.3).
Each square Monge matrix can be easily decomposed into a 
staircase Monge matrix and a flipped staircase Monge matrix.
A rectangular Monge matrix can be padded with either some number
of copies of the last row or some number of copies of the last column
in order to make it square.
Thus, for each pair $(k,l)$, $k\neq l$, we add to $D_i$ four (flipped)
staircase Monge matrices.
In total, the set $D_i$ has $2h+4h(h-1)=O(1)$ staircase Monge
matrices, each of size no more than $|\bnd{G_i}|\times|\bnd{G_i}|$.
\end{proof}

We use Lemma~\ref{l:decomp} to decompose $\ddg(G_i)$ into staircase Monge matrices.
Now, given a price function $\phi$, for each region $G_i$ we define
the set of matrices $D_i^\phi$.
For each $\mmat\in D_i$ we include
in $D_i^\phi$ a matrix $\mmat^\phi$ with the same rows and columns
as $\mmat$, such that $\mmat^\phi_{u,v}=\mmat_{u,v}+\phi(u)-\phi(v)$.
It is easy to notice that the terms $\phi(*)$ do not influence
the Monge property.
We stress that the set $D_i^\phi$ is not computed explicitly
and the preprocessing of $\ddg(G_i)$ happens only
once, immediately after $\ddg(G_i)$ is created.
Each entry of some matrix of $D_i^\phi$ can be obtained
from the entry of the corresponding matrix in $D_i$ in constant time.

We now show how Dijkstra's algorithm can be simulated on $(\bnd{G},P)\cup\bigcup \ddg(G_i)^\phi$ with
reduced lengths in order to compute lengths of shortest
paths from the source vertex $s\in \bnd{G}$ to $\bnd{G}$,
using the data structure developed in Section~\ref{s:stair_min}.

\newcommand{\dd}{\mathcal{D}}
Recall that Dijkstra's algorithm run from the source $s$ in
graph $G=(V,E)$ grows a set $S$ of \emph{visited}
vertices of $V$ such that the lengths $d(v)$ of the shortest paths $s\to v$
for $v\in S$ are already known.
Initially $S=\{s\}$ and we repeatedly choose a vertex $y\in V\setminus S$
such that the value $z(y)=\min_{x\in S}\{d(x)+\ell(x,y):(x,y)\in E\}$
is the smallest.
$y$ is then added to $S$ with $d(y)=z(y)$.
The vertices $y\in V\setminus S$ are typically stored
in a priority queue with keys $z(y)$,
which allows to choose the best $y$ efficiently.

Our implementation (see Algorithm~\ref{alg:frd}) also maintains a growing subset $S\subseteq \bnd{G}$ of visited vertices
and the values $\off(x)$ for $x\in S$.
We build a data structure of Lemma~\ref{l:staircase_report}, for each matrix
in $\bigcup_{i=1}^g \dd_i^\phi$.
$\dd_i^\phi$ is an ,,offset'' version of $D_i^\phi$.
For each matrix $\mmat\in D_i^\phi$, there is a corresponding
offset matrix $\mmat'=\om(\mmat,\off)$ in $\dd_i^\phi$ with the same rows and columns.
The row $u$ of $\mmat\in \dd_i^\phi$ is activated
(see Section~\ref{s:stair_min})
immediately once $u$ is added to $S$.
By Fact~\ref{f:offset}, the matrices of $\dd_i^\phi$ are Monge matrices.
The matrices~$\dd_i^\phi$ are never stored explicitly.
Each entry is computed from the corresponding
entry in $D_i$, the price function $\phi$ and the array $\off$
in $O(1)$ time every time it is accessed.

\newcommand{\hkey}{key}

We have a priority queue $H$ storing an element $x$ for each $x\in \bnd{G}\setminus S$.
Denote by $\hkey(e)$ the key of an element $e\in H$.
Let $W_{\mrow}$ ($W^\mcol$ respectively) be the set of all matrices of $\bigcup \dd_i^\phi$
with row $\mrow$ (column $\mcol$).
For some data structure $\mmat$ of Lemma~\ref{l:staircase_report}, denote
by $\mcols^*(\mmat)$ the set of columns of $\mmat$, for which
the minima have been already reported.

In our algorithm, we cannot afford to set $\hkey(y)$
for each $y\in \bnd{G}\setminus S$ to
$$z(y)=\min_{x\in S}\left\{\min\{d(x)+\ell^\phi(x,y) :(x,y)\in P\},\
\min\{\mmat_{x,y} : \mmat\in W_x\cap W^y\}\right\},$$
as would Dijkstra's algorithm do.
Instead, for $y\in \bnd{G}\setminus S$, $\hkey(y)$ satisfies
$$\hkey(y)=\min_{x\in S}\left\{\min\{d(x)+\ell^\phi(x,y):(x,y)\in P\},\
    \min\{\mmat_{x,y} : \mmat\in W_x, y\in\mcols^*(\mmat)\}\right\},$$
    We also add $O(g)$ special elements $\{\mmat:\mmat\in \bigcup_{i=1}^g\dd_i^\phi\}$
to our priority queue $H$.
At all times we have $\hkey(\mmat)=\mmat.\mlb()$.
We also ensure that for each $x\in S$, in every $\mmat\in W_x$,
row $x$ is activated.

The above invariants imply that for $y\in \bnd{G}\setminus S$ we have
\begin{gather}\label{e:zbound}
  z(y)\geq \min(\hkey(y), \min\{\mmat.\mlb() :\mmat\in W^y\}).
\end{gather}
Indeed, for each $x\in S$ and $\mmat\in W_x$ such that $y\notin C^*(\mmat)$,
by the definition of 
$\mmat.\mlb$, we have 
$\min\{\mmat_{x,y} : y\notin C^*(\mmat)\}\geq \mmat.\mlb()$.
Also the definition of $\hkey(y)$ implies that if $\hkey(y)\neq\infty$
then $\hkey(y)$ is the length of some $s\to y$ path.

\begin{algorithm}[htb!]
\begin{algorithmic}[1]
\Function{$\textsc{Dijkstra}$}{$P, \ell, \phi, s$}
\State Initialize the data structures of Lemma~\ref{l:staircase_report}
for each $\mmat\in \bigcup_i\dd_i^\phi$.
\State $H := $ empty priority queue
\State $S := \emptyset$
\Procedure{$\textsc{Visit}$}{$x, val$}
  \State $S := S\cup \{x\}$
  \State $d(x) := val$
  \For{$\mmat\in W_x$}
    \State $\mmat.\marow(x)$
    \State $H.\pqdec(\mmat,\mmat.\mlb())$
  \EndFor
  \For{$(x,y)\in P$}
    \If{$y\notin S$}
    \State $H.\pqdec(y,d(x)+\ell^\phi(x,y))$\label{alg:relax}
    \EndIf
  \EndFor
\EndProcedure
\For{$x\in \bnd{G}\setminus\{s\}$}
  \State $d(x)=\infty$
  \State $H.\pqins(x,\infty)$
\EndFor
\For{$\mmat\in\bigcup_i\dd_i^\phi$}
  \State $H.\pqins(\mmat,\infty)$
\EndFor
\State $\textsc{Visit}(s, 0)$
\While{$S\neq \bnd{G}$ \textbf{ and } $H.\pqminkey()\neq\infty$}\label{alg:begwhile}
\State $v := H.\pqminkey()$
\State $Z := H.\pqext()$\label{alg:extract}
\If {$Z$ is a vertex of $\bnd{G}$}
  \State $\textsc{Visit}(Z, v)$  
\Else
  \State $x = Z.\meb()$
  \If{$x\neq\nil$ \textbf{and} $x\notin S$}
    \State $H.\pqdec(x,Z.\mcolmin(x))$
  \EndIf
  \State $H.\pqins(Z,Z.\mlb())$
\EndIf
\EndWhile
\State \textbf{return} $d$
\EndFunction
\end{algorithmic}
\caption{Pseudocode of our single-source shortest paths algorithm.
The function $\textsc{Dijkstra}$ returns a vector~$d$
containing the lengths of the shortest paths from $s$ to all other
vertices of $\bnd{G}$.
We assume that $\ddg(G)$ is preprocessed so that we can access
the matrices of sets $D_1,\ldots,D_g$.}
\label{alg:frd}
\end{algorithm}

One can easily verify that the key invariants are satisfied before the first
iteration of the $\textbf{while}$ loop in line~\ref{alg:begwhile}.

Assume the element that gets extracted from $H$ in
line~\ref{alg:extract} is some vertex $x\in \bnd{G}\setminus S$.
We need to prove that $x$ has the least value $z(x)$ among
all vertices of $\bnd{G}\setminus S$ and that $z(x)=\hkey(x)$.
As the keys of $H$ include all keys $\hkey(y)$ where $y\in \bnd{G}\setminus S$
and all keys $\hkey(\mmat)$ for the $O(g)$ data structures,
for each $y\in\bnd{G}\setminus S$ we have $z(y)\geq \hkey(x)$.
But there actually exists a path $s\to x$ of length $\hkey(x)$,
so $z(y)\geq \hkey(x)= z(x)$ and thus $x$ has the minimal
$z(x)$ among all vertices in $\bnd{G}\setminus S$.
Consequently, $x$ can be safely added to~$S$.
The procedure
$\textsc{Visit}$ is used to update the set $S$, the array $d$
and all the keys of~$H$ affected by inserting $x$ to $S$.

Otherwise, the element extracted from $H$ is some data structure
$Z$.
We try to extract some previously unknown column minimum
of $Z$ with the call $Z.\meb()$.
In order to do this, we need to guarantee that all the rows
of $Z$ that are not active at that point contain
only values not less than $Z.\mlb()$.
Notice that after each extraction of an element with key $v$ from $H$,
we update some other keys in $H$ to values not less than $v$.
Thus, each extracted element has key not less than the
previously extracted elements.
In particular, we know that for each $y\in \bnd{G}\setminus S$,
$d(y)\geq Z.\mlb()$.
For each $\mmat\in W_y$, the values in row $y$ are not less
than $d(y)$ and indeed $x=Z.\meb()$ can be called.
If $x\neq\nil$, a column minimum of $Z$ has been found
and the key of $x$ is updated.
Finally, $Z$ is reinserted into $H$ with the key equal to the
new value of $Z.\mlb()$.

Let us now bound the running time of the function $\textsc{Dijkstra}$.
By Lemma~\ref{l:staircase_report}, the initialization along
with any sequence of operations $\marow$ and $\mbeb$ can be performed on 
$\mmat\in \dd_i^\phi$ in
$O\left(|\bnd{G_i}|\left(\frac{\log{|\bnd{G_i}|}}{\log\log{|\bnd{G_i}|}}\right)^2\right)$
time.

The time spent on extracting elements from $H$ is $O(I\log{|\bnd{G}|})$,
where $I$ is the number of insertions into $H$.
Clearly $H$ never contains more than $O(|\bnd{G}|+g)=O(|\bnd{G}|)$ elements.
Each vertex of $\bnd{G}$ is inserted into $H$ at most once
and, by Lemma~\ref{l:staircase_report}, each data structure $\mmat\in \dd_i^\phi$
is inserted into $H$ at most $O(|\bnd{G_i}|\log^{\epsilon}{|\bnd{G_i}|})$ times
before it reports all the column minima.
Hence, the total time spend on $H.\pqext$ is
$O\left(\log{|\bnd{G}|}\sum_i|\bnd{G_i}|\log^{\epsilon}{|\bnd{G}|}\right)$.
The operation $H.\pqdec$ takes constant time and
thus we can neglect the calls to $H.\pqdec$ immediately
after $\marow$ or $\mblb$.
However, there are also $O(|P|)$ calls to $H.\pqdec$ in
line~\ref{alg:relax}, which cannot be neglected this way.
Taking into account the preprocessing of Lemma~\ref{l:decomp}, we conclude
with the following theorem.
\begin{theorem}\label{t:dijkstra}
  Let $\epsilon\in (0,1)$. After preprocessing $\ddg(G)$ in
  $O\left(\sum_i(|V(G_i)|+|\bnd{G_i}|^2)\log{|V(G_i)|}\right)$ time, 
  an instance of the single-source shortest paths problem on $\ddg(G)$ can be solved in \linebreak
  $O\left(\sum_i|\bnd{G_i}|\left(\frac{\log^2{|\bnd{G_i}|}}{\log^2\log{|\bnd{G_i}|}}+\
    \log^{\epsilon}{|\bnd{G_i}|}\log{|\bnd{G}|}\right)+|P|\right)$
time.
\end{theorem}

\begin{remark}
Currently the most efficient known algorithm for computing $\ddg(G_i)$ given a feasible
price function $\phi_i^0$ \cite{Klein:2005} runs in
$O\left((|V(G_i)|+|\bnd{G_i}|^2)\log{|V(G_i)|}\right)$ time,
which is asymptotically the same as the time needed for preprocessing
in Theorem~\ref{t:dijkstra}.
Consequently, we neglect the preprocessing time in Theorem~\ref{t:dijkstra}
when discussing applications.
\end{remark}

\section{Implications}\label{s:impli}
The implications of Theorem~\ref{t:dijkstra} are numerous.
In this section we mention some planar graph problems for which the best
known algorithms compute single-source shortest paths in dense
distance graphs and such step constitutes the main bottleneck
of their running times.

\paragraph{Multiple-Source Multiple-Sink Maximum Flow in Directed Planar Graphs}
In this problem we are given a directed planar graph $G=(V,E)$ with real edge capacities.
Let $n=|V|$.
We are also given two subsets $S,T\subseteq V$, $S\cap T=\emptyset$.
The set $S$ contains \emph{source vertices}, while $T$ contains \emph{sink vertices}.
Our goal is to send as much flow from the vertices $S$ to vertices of $T$
without violating edge capacity constraints and flow conservation
on the vertices $V\setminus S\setminus T$.
Although in general graphs this problem can be reduced to
single-source single-sink maximum flow by adding a super-source and
a super-sink (connected with all sources and all sinks, respectively),
such a reduction does not
work for planar graphs as it does not preserve planarity.
Note that the problem of computing maximum matching in a bipartite
planar graph can be reduced to a single-source single-sink maximum
flow instance.

Borradaile et al. \cite{Borradaile:2011} found an $O(n\log^3{n})$
algorithm for this problem.
Their algorithm recursively partitions $G$ in a balanced way with cycle
separators of size $C=O(\sqrt{n})$. The results of
recursive calls are combined using $O(C)$
computations of single-source shortest paths
in a single-region dense distance graph with boundary size
$O(C)$ and $O(C)$ additional edges.
The implementation they use runs in $O(C\log^2{C})$ time.
If we replace it with the implementation of Theorem~\ref{t:dijkstra},
we obtain the following corollary.

\begin{corollary}
The multiple-source multiple-sink maximum flow and the maximum bipartite matching
in a planar graph can
be computed in $O\left(n\frac{\log^3{n}}{\log^2\log{n}}\right)$ time.
\end{corollary}

\paragraph{Single-Source All-Sinks Maximum Flow in Directed Planar Graphs}
In this problem we are also given a directed planar graph $G=(V,E)$
with real edge capacities
and some $s\in V$.
Our goal is to compute the values of the maximum flow between $s$
and all vertices $t\in V\setminus\{s\}$.

Łącki et al. \cite{Lacki:2012} gave an $O(n\log^3{n})$-time algorithm
for this problem.
The overall structure of their algorithm is similar to this
of Borradaile et al. \cite{Borradaile:2011} and the bottleneck
on each level of the recursive decomposition is to solve
$O(X)$ instances of a single-source shortest path problem
in a dense distance graph with a total boundary of $O(X)$,
where $X=O(\sqrt{n})$.

\begin{corollary}
The single-source all-sinks maximum flow in planar graphs can
be solved in 
$O\left(n\frac{\log^3{n}}{\log^2\log{n}}\right)$ time.
\end{corollary}

\paragraph{Exact Distance Oracles for Directed Planar Graphs}
Mozes and Sommer \cite{Mozes:2012} considered the following problem.
Given a planar digraph $G=(V,E)$ with real edge lengths
and space allocation $S\in [n\log\log{n},n^2]$,
construct a data structure of size $O(S)$ answering exact distance
queries in $G$ as efficiently as possible.
They proposed a data structure that can be constructed
in $O\left(S\frac{\log^3{n}}{\log\log{n}}\right)$ time,
and is capable of answering the distance queries in
$O\left(\frac{n}{\sqrt{S}}\log^2{n}\log^{3/2}{\log{n}}\right)$ time.
At the heart of their query algorithm lies the basic
version of FR-Dijkstra (without reduced costs), and thus
replacing it with our algorithm gives a faster query algorithm for $S=O(n^{2-\epsilon})$,
for any $\epsilon>0$.
\begin{corollary}
  Given a planar graph $G$ and $S\in [n\log{\log{n}},n^{2-\epsilon}]$,
  one can construct a $O(S)$-space exact distance oracle for $G$
  in $O\left(S\frac{\log^3{n}}{\log\log{n}}\right)$ time.
  The oracle answers shortest path queries
  $O\left(\frac{n}{\sqrt{S}}\frac{\log^2{n}}{\log^{1/2}\log{n}}\right)$
  time.
\end{corollary}

\paragraph{Fully-Dynamic Distance Oracles for Directed Planar Graphs}
In this problem we are given a directed planar graph $G$ with real edge lengths
which undergoes edge insertions and deletions.
It is also guaranteed that edge insertions do not break the planarity
of $G$.
Italiano et al. \cite{Italiano:2011} showed a fully dynamic data
structure limited to the case of non-negative edge lengths.
On the other hand, Kaplan et al. \cite{Kaplan:2012} showed a data structure processing
updates and answering queries in $O(n^{2/3}\log^{5/3}{n})$ time
in the case, when only edge-length updates are allowed.
Both data structures can be easily combined in order to allow
both edge set updates and negative lengths.
Again, FR-Dijkstra on a dense distance graph induced by an
$r$-division can be identified as one of the bottlenecks
of both query and update algorithms.
The second bottleneck is the computation of a dense distance
graph of a piece using the data structure of Klein \cite{Klein:2005}
in $O(r\log{r})$ time.
The terms $O\left(\frac{n}{\sqrt{r}}\frac{\log^2{n}}{\log^2\log{n}}\right)$ and $O(r\log{n})$ can
be balanced for $r=n^{2/3}\frac{\log^{2/3}}{\log^{4/3}\log{n}}$.
\begin{corollary}
  For a directed planar graph $G$, in $O\left(n\frac{\log^{2}}{\log\log{n}}\right)$ we can construct
  a data structure supporting both edge updates (insertions and deletions) and distance queries
  in $O\left(n^{2/3}\frac{\log^{5/3}}{\log^{4/3}\log{n}}\right)$ amortized time.
\end{corollary}

\subsection*{Acknowledgements}
We thank Piotr Sankowski for helpful discussions. The first author also thanks Oren Weimann and Shay Mozes for discussions about Monge matrices and FR-Dijkstra.

\bibliography{paper}

\end{document}